\documentclass{article}
 \usepackage{xcolor}
\definecolor{blue(ryb)}{rgb}{0.01, 0.28, 1.0}
\definecolor{brandeisblue}{rgb}{0.0, 0.44, 1.0}
\definecolor{ceruleanblue}{rgb}{0.16, 0.32, 0.75}
\definecolor{cobalt}{rgb}{0.0, 0.28, 0.67}
\definecolor{coolblack}{rgb}{0.0, 0.18, 0.39}
\definecolor{darkblue}{rgb}{0.0, 0.0, 0.55}
\usepackage[hidelinks]{hyperref}
\hypersetup{
    colorlinks,
    citecolor=darkblue,
    filecolor=black,
    linkcolor=darkblue,
    urlcolor=black
}
\usepackage{amsmath, amssymb, amsfonts, amsthm}
\usepackage{graphicx}
\def\nz{\ifmmode {I\hskip -3pt N} \else {\hbox {$I\hskip -3pt N$}}\fi}
\def\zz{\ifmmode {Z\hskip -4.8pt Z} \else

       {\hbox {$Z\hskip -4.8pt Z$}}\fi}
\def\qz{\ifmmode {Q\hskip -5.0pt\vrule height6.0pt depth 0pt
       \hskip 6pt} \else {\hbox
       {$Q\hskip -5.0pt\vrule height6.0pt depth 0pt\hskip 6pt$}}\fi}
\def\rz{\ifmmode {I\hskip -3pt R} \else {\hbox {$I\hskip -3pt R$}}\fi } 
\def\cz{\ifmmode {C\hskip -4.8pt\vrule height5.8pt \hskip 6.3pt} \else 
{\hbox {$C\hskip -4.8pt\vrule height5.8pt \hskip 6.3pt$}}\fi} 

\def\curl{{\rm curl}\,}

\def \Ab{{\bf A}}

\def\and {{\rm \; and \;}}

\def\R{\mathbb R}
\def\N{\mathbb N}
\def\Z{\mathbb Z}

\definecolor{ao(english)}{rgb}{0.0, 0.5, 0.0}

\newcommand{\beq}{\begin{equation}}
\newcommand{\eeq}{\end{equation}}
\makeatletter
\@addtoreset{equation}{section}

\makeatother

\newtheorem{theorem}{Theorem}[section]

\newtheorem{lemma}[theorem]{Lemma}
\newtheorem{proposition}[theorem]{Proposition}

\newtheorem{remark}[theorem]{Remark}
\newtheorem{corollary}[theorem]{Corollary}

\title{Quantum tunneling in  deep potential wells and  strong magnetic field revisited}

\author{Bernard Helffer$^1$ and  Ayman  Kachmar$^{2,3}$}

\begin{document}
\bibliographystyle{plain}
\maketitle 
$^1$ Laboratoire de Math\'ematiques Jean Leray, CNRS,  Nantes Universit\'e, 44000 Nantes, France.  Email: bernard.helffer@univ-nantes.fr\\

$^2$ Lebanese University,  Faculty of Sciences, 1700  Nabatieh, Lebanon.

$^3$ Current address: 

The Chinese University of Hong Kong (Shenzhen), School of Science and Engineering, 2001 Longxiang Blvd., Longgang District, Shenzhen, China.  

Email: akachmar@cuhk.edu.cn

\begin{abstract}
Inspired by a recent paper$^*$ by C. Fefferman, J. Shapiro and M. Weinstein, we  investigate quantum tunneling for a Hamiltonian with a symmetric  double well and a uniform  magnetic  field.  In the simultaneous limit of  strong magnetic field  and   deep  potential wells  with disjoint supports,   tunneling occurs and we derive   accurate estimates of its magnitude.  \medskip

$^*$\,[Lower bound on quantum tunneling for strong magnetic fields. {\it 
SIAM J. Math. Anal.} 54(1), 1105-1130 (2022).]
\end{abstract}

\section{Introduction}

\subsection{The Hamiltonian}

\subsubsection{The double well potential}

Consider $\mathfrak v_0\in C_c^\infty(\R^2)$   such that
\begin{equation}\label{eq:v0}
\begin{cases}
\mathfrak v_0(x)=v_0(|x|)~{\rm is ~radial}~\&~v_0^{\min}:=\min\limits_{r\geq 0}v_0(r)<0\,,\\
{\rm supp}\,\mathfrak v_0\subset  \overline{D(0,a)}:=\{x\in\R^2\,:\,|x|\leq a\}\,,\\
U_0:=\{\mathfrak  v_0(x)=v_0^{\min}\}=\{0\}\quad\&\quad   v_0''(0)>0\,.
\end{cases}
\end{equation}
 We suppose that $\overline{D(0,a)}$ is the smallest  disc containing ${\rm supp}\,\mathfrak v_0$, i.e.
\begin{equation}\label{eq:def-a}
a=  a(\mathfrak v_0):=\inf\{r>0~:~{\rm  supp}\,\mathfrak v_0\subset D(0,r)\}\,.
\end{equation}
 We introduce the \emph{double well}  potential ($\ell$ refers to `left' and $r$ to `right'),
\begin{equation}\label{eq:V}
V(x)=\mathfrak v_0(x-z^{\ell})+\mathfrak v_0(x-z^r)\quad{\rm where~}|z^\ell-z^r|=:L>2a(\mathfrak v_0)\,.
\end{equation}
The potential wells of $V$ associated with the energy $v_0^{\min}$ are the connected components of $\{V(x)=\min V\}$, i.e.
\begin{equation}\label{eq:V-wells}
U_\ell=U_0+z^\ell=\{z^\ell\}\,,\quad U_r=U_0+z^r=\{z^r\}\,.
\end{equation}
Without loss of generality, we  choose  $z^\ell$  and  $z^r$  in  the following manner
\begin{equation}\label{eq:z-lr}
z^\ell=\Big(-\frac{L}2,0\Big),\quad z_r=\Big(\frac{L}2,0\Big)\,.
\end{equation}

\subsubsection{The magnetic field}

Consider a  vertical magnetic field $b\vec{z}$ where
\begin{equation}\label{eq:B}
b>0 {\rm ~is~a~constant.}
\end{equation}
Notice that
\[b=\curl(b\Ab)\]
where $\Ab$ is defined in  polar coordinates $(r,\theta)$  as  follows,
\begin{equation}\label{eq:A}
\Ab(r,\theta)=\frac{r}2\left[ \begin{array}{r}-\sin\theta\\ \cos\theta\end{array} \right]\,.
\end{equation}
\subsection{Deep symmetric wells in a strong magnetic field}
We consider the Hamiltonian 
\begin{equation}\label{eq:H-lambda}
\mathcal H_{b,\lambda}:=(D-b\Ab)^2+  \lambda^2 V\,,\quad D:=\frac{1}{i}\nabla \,,
\end{equation}
with a double well electric potential $\lambda^2V$ and a magnetic potential $b \Ab$, where  $\lambda$ and  $b$  are the  coupling parameter and  the intensity of the magnetic field, respectively. In this  paper, we suppose that  $b=\lambda$ and $\lambda\gg 1$ is large\footnote{Writing $\lambda\gg 1$ means that we consider the regime where $\lambda\to+\infty$.  In the same vein, writing $\alpha\ll\lambda$ (resp. $\alpha\gg \lambda$), we mean that $\alpha/\lambda\to0$ (resp.   $\alpha/\lambda\to+\infty$).}.

The regime where $b$ does not scale like the coupling parameter  $\lambda$ has been inspected a long time ago. For instance, when   $b\ll\lambda$, accurate estimates of the tunnel effect where obtained in \cite{HeSjPise}, while when $b\gg\lambda$, the effect of the potential well becomes weak  and the magnetic effect is dominant (see \cite{Be} and \cite{HeSjSond}).

The potential function considered  in \eqref{eq:H-lambda} is not analytic, thereby making our setting   significantly different from the one of \cite{HeSjPise}, where the magnetic field scales like the coupling parameter, $V$ is analytic  and some condition on the intensity of the magnetic field appears. As we shall see, this will induce serious difficulties in deriving accurate bounds on the magnitude of the  tunnel effect and highlights another interesting new phenomenon related to \emph{tunneling} under a magnetic field compared to  recent results in \cite{BHR, FHK}.

In order to exploit the connection with the rich literature on the tunnel effect in  multiple wells (see \cite{DS, He88, HeSj1, HeSjPise}),  it will be convenient to divide by $\lambda$ and mainly consider the corresponding equivalent semi-classical problem 
\begin{equation}\label{eq:H-h}
\mathcal L_h:=(hD-\Ab)^2+V\,,
\end{equation}
where $ h=\lambda^{-1}\ll 1$.\\
Hence we have
\[
\mathcal H_{b,\lambda}=h^{-2}\mathcal L_h\,.
\]
Our result will depend on the size of the support of  the potential  function  (in particular through $a(\mathfrak v_0)$ in \eqref{eq:def-a}). Let us denote by $(e_j^{\mathfrak v_0}(h))_{j\geq 1}$  the sequence of min-max eigenvalues of  $\mathcal L_h$. We will investigate the semi-classical asymptotics of 
\begin{equation}\label{eq:gap}
e_2^{\mathfrak v_0}(h)-e_1^{\mathfrak v_0}(h)\,,
\end{equation} 
and prove roughly speaking (see Corollary~\ref{corol:maina} for a precise statement)
\begin{equation}\label{eq:weak-asy}
e_2^{\mathfrak v_0}(h)-e_1^{\mathfrak v_0}(h) \underset{\substack{h\to0\\a(\mathfrak v_0)\to0}}{\sim} \exp\left(-\frac1h\int_0^{L}\sqrt{\frac{\rho^2}{4}-v_0^{\min}}\,d\rho\right)\,.
\end{equation} 
 Under the additional assumption  that $\mathfrak v_0$ does not vanish in  the open disk  $D\big(0,a(\mathfrak v_0)\big)$, we prove 
an accurate asymptotics of the form (see Theorem~\ref{thm:main*})
\[e_2^{\mathfrak v_0}(h)-e_1^{\mathfrak v_0}(h) \underset{h\to0}{=}\exp\left(-\frac{S(\mathfrak v_0)+o(1)}{h}\right)\]
without the hypothesis that $a(\mathfrak v_0)\ll 1$. 

Our  investigation relies on expanding the ground state $e^{\rm sw}(h)$ of the single well Hamiltonian
\begin{equation}\label{eq:1well}
\mathcal L_h^{\rm sw}:=(hD-\Ab)^2+\mathfrak v_0\,.
\end{equation}
Under the assumptions in \eqref{eq:v0}, we show that:
\begin{theorem}[ Existence of radial ground states and precise expansions]\label{thm:main0}~

Assume that $\mathfrak v_0$ satisfies the conditions in \eqref{eq:v0}. Then, there exists $h_0>0$ such that, for all $h\in(0,h_0]$, the following holds:
\begin{enumerate}
\item The ground  state energy, $e^{\rm sw}(h)$, of $\mathcal L_h^{\rm sw}$, is a simple eigenvalue and
\begin{equation}\label{eq:gs}
e^{\rm sw}(h)= v_0^{\min}+h\sqrt{1+2v_0''(0)}+\mathcal O (h^{3/2}) \,.
\end{equation}
\item 
 There exists a unique positive ground state, $\mathfrak u_h$, of $\mathcal L_h^{\rm sw}$, with the properties
 \begin{itemize}
 \item $\mathfrak u_h(x)=u_h(|x|)$ is a radial function\,;
 \item $\mathfrak u_h$ is normalized, i.e. $\int_{\R^2}|\mathfrak u_h|^2dx=1$\,.
 \end{itemize} 
 \item  There exists  a positive radial function $\mathfrak a_0$ on $\mathbb R^2$  satisfying 
 \begin{equation}\label{eq:defa0}
 \mathfrak a_0(0)=\frac1{2}\frac{\sqrt{1+2v_0''(0)}}{\pi}\,,
 \end{equation}
  such that, for any $R>0$,  the ground state $\mathfrak u_h$ satisfies, uniformly in the disc $D(0,R)\subset\R^2$,
 \begin{equation}\label{eq:uh-main0}
  \left|e^{\mathfrak d(x)/h}\mathfrak  u_h(x)-h^{-1/2}\mathfrak a_0(x) \right|=\mathcal O(h^{1/2})\,,
  \end{equation}
 where 
 \begin{equation}\label{eq:d-main0}
 \mathfrak d(x)=d(|x|)=\int_0^{|x|}\sqrt{\frac{\rho^2}{4}+v_0(\rho)-v_0^{\min}}\,d\rho\,.
 \end{equation}
\end{enumerate}
\end{theorem}
This in particular clarifies the hypotheses imposed in \cite{FSW}. Then, applying 
Theorem~1.5 in \cite{FSW},  under the  additional  assumptions
\begin{equation}\label{assFSW}
\mathfrak v_0\leq 0 \mbox{ and }L>4\left(\sqrt{|v_0^{\min}|}+a(\mathfrak v_0)\right)\,, 
\end{equation}
we get that
\begin{equation}\label{eq:FSW}
\begin{aligned}
 \exp\left(-\frac{L^2+4\sqrt{|v_0^{\min}|}L+\gamma(\mathfrak v_0)}{4h} \right)& \leq e_2^{\mathfrak v_0}(h)-e_1^{\mathfrak v_0}(h) \\
 &\leq \exp\left(-\frac{(L-a(\mathfrak v_0))^2-a(\mathfrak v_0)^2}{4h} \right) 
 \end{aligned}\end{equation}
where $\gamma(\mathfrak v_0)$ is a positive constant. 

 The bounds in \eqref{eq:FSW} follow from the asymptotics  (also obtained in \cite[Eq.~(1.12)]{FSW} under the assumptions \eqref{eq:FSW})
\begin{equation}\label{eq:FSW-main}
e_2^{\mathfrak v_0}(h)-e^{\mathfrak v_0}_1(h)\underset{h\to0}{\sim} \Big|2\int_{D(0,a)} \mathfrak v_0(x)\mathfrak u_h(x)\mathfrak u_h(x_1+L,x_2)e^{\frac{iLx_2}{2h}}\,dx\, \Big|
\end{equation}
where $\mathfrak u_h$ is the radial  ground  state of $\mathcal L_h^{\rm sw}$ (see Theorem~\ref{thm:main0}). The integral in the right hand side of \eqref{eq:FSW-main} is called, after \cite{FSW},    the \emph{hopping coefficient.} It describes the interaction between the two potential ``wells'' and can be derived through a reduction  to the restriction of $\mathcal L_h$  on a two dimensional space (yielding an \emph{interaction matrix} like in \cite{He88}).

Using the improved expansion of the ground state $\mathfrak u_h$ in Theorem~\ref{thm:main0} above,  we improve the bounds on the hopping coefficient and thereby on $e_2^{\mathfrak v_0}(h)-e_1^{\mathfrak v_0}(h)$ provided the potential $\mathfrak v_0$ satisfies the conditions in \eqref{eq:v0}. 

 Besides its role in capturing the tunneling asymptotics, precise estimates of the hopping coefficient 
are key ingredients in the understanding of tight binding reductions \cite{SW} (see also \cite{D, O}). Hence  our new estimates of the hopping have  potential applications, to be discussed elsewhere.

Our  first result, on  the eigenvalue splitting, is as follows.

\begin{theorem}[New bounds on the eigenvalue splitting]\label{thm:main}~

 Assume that  $\mathfrak v_0 $  satisfies  the conditions  in \eqref{eq:v0} and  \eqref{assFSW}.
Then, there exist positive constants $h_0,C_1,C_2$    such that    for all $h\in(0,h_0)$, we have
\begin{equation}\label{eq:main}
 C_1h\exp\left(-\frac{S_0 }{h}\right)\leq e_2^{\mathfrak v_0}(h)-e_1^{\mathfrak v_0}(h)\leq C_2h^{-1}\exp\left(-\frac{S_a}{h}\right)
\end{equation}
where (with $a=a(\mathfrak v_0)>0$)
\begin{equation}\label{eq:Sa}
S_a=\int_0^{L-a} \sqrt{\frac{\rho^2}{4}+v_0(\rho)-v_0^{\min}}\,d\rho  + \int_0^a\sqrt{\frac{\rho^2}{4}+v_0(\rho)-v_0^{\min}}\,d\rho
\end{equation}
and
\begin{equation}\label{eq:S0-main}
S_0=\int_0^L \sqrt{\frac{\rho^2}{4}+v_0(\rho)-v_0^{\min}}\,d\rho \,.
\end{equation}
Furthermore, if $\mathfrak v_0<0$ in $D(0,a)$, then we have the improved lower bound
\begin{equation}\label{eq:main*}
\liminf_{h\to0} h\ln \big(e_2^{\mathfrak v_0}(h)-e_1^{\mathfrak v_0}(h)\big)\geq - \hat S\,,
\end{equation}  
where $\hat S<S_0$ is defined as follows
\begin{equation}\label{eq:hat-S}
\hat S=\inf_{0<r<a}\left(\frac{Lr}{2}+\int_0^{L-r}\sqrt{\frac{\rho^2}{4}+v_0(\rho)-v_0^{\min}}d\rho+\int_0^r \sqrt{\frac{\rho^2}{4}+v_0(\rho)-v_0^{\min}}d\rho\right)\,.
\end{equation} 
\end{theorem} 
Notice that, whenever $L>2a$ and $v_0\leq 0$, we have,
\begin{equation}\label{eq:potential} 
\frac{\rho}{2}\leq  \sqrt{\frac{\rho^2}{4}+v_0(\rho)-v_0^{\min}}\leq \sqrt{\frac{\rho^2}{4}-v_0^{\min}}\leq \frac{\rho}{2}+\sqrt{|v_0^{\min}|}\,, 
\end{equation}
from which  we deduce (with $a=a(\mathfrak v_0)$) the following estimates
\begin{equation} \frac{(L-a)^2-a^2}{4}< S_a~{\rm and~}
 S_0<  \frac{L^2+4\sqrt{|v_0^{\min}|}L}{4}\,.
 \end{equation}
 Hence the estimates  in  \eqref{eq:main} already improve the ones in \eqref{eq:FSW} thereby providing a more   accurate measurement of the magnitude of the  tunnel effect.
The  guess of $S_a$ and $S_0$  is based on the consideration of the hopping matrix like it appears in \cite{FSW} together with the WKB expansion of the single  well ground states obtained in  the present contribution (Theorem~\ref{thm:main0}). 
 The estimate \eqref{eq:main*} is a new improvement of the lower bound in  \eqref{eq:main} since we will  prove  in Proposition~\ref{prop:hat-S}, that $S_a<\hat S<S_0$.

The bounds in Theorem~\ref{thm:main} provide   a rather sharp estimate of the gap $e_2^{\mathfrak v_0}(h)-e_1^{\mathfrak v_0}(h)$ when the size of the support of  the potential function is small (i.e. $a(\mathfrak v_0)\ll1$). In fact, we have 
\begin{equation}\label{eq:S0=Sa+Ra}
 S_0=S_a+R_a
 \end{equation}
where
\[R_a:=
\int_0^a\left(\sqrt{\frac{(L-\rho)^2}{4}-v_0^{\min}}\,-\sqrt{\frac{\rho^2}{4}+v_0(\rho)-v_0^{\min}}\right)d\rho\,.\]
If $v_0\leq 0$ and $0<2a<L$, then $R_a$ satisfies (for the upper bound we use \eqref{eq:potential})
\begin{equation}\label{eq:Ra} 
0< R_a\leq \left(\frac{L-a}{2}+\sqrt{|v_0^{\min}|}\right)a\,.
\end{equation}
Moreover, observing that
\[
S_0 =\int_0^L\sqrt{\frac{\rho^2}{4}-v_0^{\min}}\,d\rho+\int_0^a\left(\sqrt{\frac{\rho^2}{4}+v_0(\rho)-v_0^{\min}}-\sqrt{\frac{\rho^2}{4}-v_0^{\min}}\right)d\rho
\]
we obtain (when  $v_0\leq 0$ and for the  lower bound we use \eqref{eq:potential})
\[ \int_0^L\sqrt{\frac{\rho^2}{4}-v_0^{\min}}\,d\rho-  \sqrt{|v_0^{\min}|}\,a\leq  S_0\leq  \int_0^L\sqrt{\frac{\rho^2}{4}-v_0^{\min}}\,d\rho\,. \] 
We then have the following immediate corollary of Theorem~\ref{thm:main} (which is our precise meaning of 
\eqref{eq:weak-asy}).
\begin{corollary}\label{corol:maina}
Under the assumptions in Theorem~\ref{thm:main}, the following holds
\begin{equation*}
\begin{aligned} 
-\int_0^L\sqrt{\frac{\rho^2}{4}-v_0^{\min}}\,d\rho  &\leq 
 \liminf_{h\to0_+} h\ln\big(e_2^{\mathfrak v_0}(h)-e_1^{\mathfrak v_0}(h)\big)\\
 &\leq 
\limsup_{h\to0_+} h\ln\big(e_2^{\mathfrak v_0}(h)-e_1^{\mathfrak v_0} (h)\big)\\
& \leq -\int_0^L\sqrt{\frac{\rho^2}{4}-v_0^{\min}}\,d\rho+C_L(v_0) \,,
\end{aligned}
\end{equation*}
where
\[C_L(\mathfrak v_0)=\left(\frac{L-a(\mathfrak v_0)}{2}+2\sqrt{|v_0^{\min}|}\right)a(\mathfrak v_0)\]
and $a(\mathfrak v_0)$ is introduced in \eqref{eq:def-a}.   
\end{corollary}

 With the improved bound in \eqref{eq:main*} holding when $\mathfrak v_0<0$ in $D\big(0,a(\mathfrak v_0)\big)$, we can refine our estimates of the \emph{hopping coefficient} on the r.h.s. of \eqref{eq:FSW-main}. The  idea is that we insert the profile of $\mathfrak u_h$ given in Theorem~\ref{thm:main0}  into the hopping coefficient in \eqref{eq:hat-w} and control  the arising  error terms by \eqref{eq:main*} and a tricky identity from \cite{FSW}  refined by our expansion of  $\mathfrak u_h$ in \eqref{eq:uh-main0} (see Lemma~\ref{lem:FSW-uh}). The outcome is a precise asymptotics of the tunneling's magnitude  through a new sharp constant $\mathcal  S(\mathfrak v_0,L)$ that  we will introduce  later in \eqref{eq:def-S(v0)**}. More precisely we prove the following.

\begin{theorem}[Sharp asymptotics of the eigenvalue splitting]\label{thm:main*}~

Under the assumptions in Theorem~\ref{thm:main}, if   $\mathfrak v_0<0$ in $D\big(0,a(\mathfrak v_0)\big)$, then we have
\[
h \ln\big(e_2^{\mathfrak v_0}(h)-e_1^{\mathfrak v_0}(h)\big) \underset{h\to0}{\sim} - S(\mathfrak  v_0,L)\,, \]
where $S(\mathfrak v_0,L)$ is a positive constant.
\end{theorem}
The expression of $S(\mathfrak v_0,L)$,  explicitly given in \eqref{eq:def-S(v0)**}, has the form
\[S(\mathfrak v_0,L)=S_a+I\big(a(\mathfrak v_0),L,v_0^{\min}\big)\]
where we see two types of  terms: 
\begin{itemize}
\item $S_a$, introduced in \eqref{eq:Sa}, is expressed in terms of  the  magnetic field, the potential function $\mathfrak v_0$, and the size of its support, $a(\mathfrak v_0)$ (see \eqref{eq:def-a}).  It   is obtained from the  magnetic Agmon distance between the wells.

\item  $I\big(a(\mathfrak v_0),L,v_0^{\min}\big)$ involves,  in addition to the size of the well's support,   the distance $L$ between the wells and  the minimum $v_0^{\min}$ of the potential, which amounts to the ``deepness'' of the well.  We can interpret $ I\big(a(\mathfrak v_0),L,v_0^{\min}\big)$ as an interaction term between the wells.
\end{itemize} 
\begin{remark}[Distant wells and narrow wells]\label{rem:interaction}
If the distance $L$ between the wells is very large, we observe that 
\[S(\mathfrak v_0,L)\underset{L\to+\infty}{=}\frac{L^2}{4}+|v_0^{\min}|\ln L+ \mathcal O(1)\] 
Notice that, to  leading order, the foregoing asymptotics is consistent with \eqref{eq:FSW}.

Even when the support of the well is very small, we find
\[S(\mathfrak v_0,L)\underset{a\to0_+}{\sim}\frac{L}4\sqrt{L^2+4|v_0^{\min}|}+|v_0^{\min}|\ln\left( \frac{L\Big(1+\sqrt{1+\frac{4|v_0^{\min}|}{L^2}}\Big)}{2\sqrt{|v_0^{\min}|}} \right)\,.\]
Notice that
\[\begin{aligned}
\int_0^L\sqrt{\frac{\rho^2}{4}-v_0^{\min}}\,d\rho&=\frac{L}{4}\sqrt{L^2+4|v_0^{\min}|}+|v_0^{\min}|\sinh^{-1}\left(\frac{L}{2\sqrt{|v_0^{\min}|}} \right)
\\
&=\frac{L}4\sqrt{L^2+4|v_0^{\min}|}+|v_0^{\min}|\ln\left( \frac{L\Big(1+\sqrt{1+\frac{4|v_0^{\min}|}{L^2}}\Big)}{2\sqrt{|v_0^{\min}|}} \right)
\end{aligned} \]
which  is consistent with Corollary~\ref{corol:maina}.
\end{remark}
\begin{remark}[Consistency with  non-magnetic tunneling]\label{rem:b=0}
By a simple change of scales, our results include the Hamiltonian
\begin{equation}\label{eq:H-h-beta}
\mathcal L_\hbar^\beta:=(\hbar D-\beta\Ab)^2+V\,,
\end{equation}
where $\beta>0$ measures the \emph{strength} of the magnetic field.\\
If we denote by  $\big(e_j(\hbar,\beta)\big)_{j\geq1}$ the sequence of min-max eigenvalues of $\mathcal L_\hbar^\beta$ and  introduce the effective semi-classical parameter $h=\beta^{-1}\hbar$, we reduce the analysis to  the Hamiltonian in \eqref{eq:H-h}, thanks to the relation
\[ \mathcal L_\hbar^\beta:=\beta^2\Big((hD-\Ab)^2+\beta^{-2}V\Big)\,. \]
This yields
\[e_2(\hbar,\beta)-e_1(\hbar,\beta)=\beta^2\left(e_2^{\beta^{-2}\mathfrak v_0}(h)-e_1^{\beta^{-2}\mathfrak v_0}(h)\right)\,. \]
As a consequence of Theorem~\ref{thm:main*}, we get
\[ \hbar \ln\big(e_2(\hbar,\beta)-e_1(\hbar,\beta)\big) \underset{\hbar\to0}{\sim}-\beta S(\beta^{-2}\mathfrak  v_0)\]
and we shall see that (cf. Proposition~\ref{prop:beta=0})
\[\beta S(\beta^{-2}\mathfrak  v_0)\underset{\beta\to0}{\sim} 2\int_0^{L/2}\sqrt{v_0(\rho)-v_0^{\min}}\,d\rho\,,\]
which is the term obtained in the non-magnetic setting \cite{HeSj1, S}.
\end{remark}
The rest of the paper is devoted to the proofs of Theorems~\ref{thm:main0}, \ref{thm:main} and  \ref{thm:main*}. 
In Section~\ref{sec:mha}, we revisit the harmonic approximation in the presence of a magnetic field and conclude by proving Proposition~\ref{prop:m-harm-app}, which proves the first and second items in Theorem~\ref{thm:main0}. In Section~\ref{sec:WKB}, we recall the WKB approximation in the setting of  radially symmetric potential and ground state; the third item in Theorem~\ref{thm:main0} then follows from Propositions~\ref{prop:WKB} and \ref{prop:WKB-app}. Section~\ref{sec:comp} is devoted to the properties of the constants $S_0$, $S_a$ and $\hat S$ appearing in the tunneling estimates \eqref{eq:main} and \eqref{eq:main*}. The proof of Theorem~\ref{thm:main} occupies Section~\ref{sec:hopping} (see  \eqref{eq:hat-w} and Propositions~\ref{prop:hop} and \ref{prop:hop*}). Finally, Section~\ref{sec:asy} is devoted to  the proof Theorem~\ref{thm:main*}.  Theorem~\ref{thm:main*} only yields the decay rate of  the tunneling  and an estimate of the amplitude. Capturing the leading term of the amplitude  seems quite challenging and will be discussed  in Remark~\ref{rem:amplitude}.
Another perspective could be to  relax the ``radial"  hypothesis  on  the potential function $\mathfrak v_0$, in particular, establishing that tunneling occurs  when $\mathfrak  v_0$ is  compactly supported with  a unique non-degenerate minimum.

\section{Magnetic harmonic approximation}\label{sec:mha}
 In the presence of a magnetic field and  a unique non-degenerate well, the method of harmonic approximation    was treated in \cite[Sec.~2]{HeSjPise}, but we revisit it here in the setting of a radial potential, which  allows us to derive more  precise results on  the ground states. We can also refer to Matsumoto \cite{Mat}  and Matsumoto-Ueki \cite{MatUek} for an independent  discussion in the general case.

\subsection{The Landau  Hamiltonian}

In the absence of an electric field, $\mathfrak v_0=0$,  the operator in \eqref{eq:1well}  reduces (after rescaling) to the Landau Hamiltonian
\[L=(D-\Ab)^2\]
whose spectrum  consists of the Landau levels, i.e.
\[\sigma(L)=\{\Lambda_n:=(2n-1)~:~n\in\N\}\,,\]
where each $\Lambda_n$ is  an eigenvalue of infinite multiplicity.\\
 Moreover, $ L$ has  a normalized radial ground state given  by
\[\phi_0(x)=\pi^{-1/2}\exp\left(-\frac{|x|^2}{2}\right)\,.\]
We can decompose $L$ via the orthogonal projections on  the Fourier modes,   
\begin{equation}\label{eq:proj-Fm}\Pi_m u:=e^{im\theta}\pi_mu,\quad \pi_mu:=\frac1{2\pi}\int_0^{2\pi} u(r,\theta')e^{-im\theta'}d\theta'\qquad(m\in\Z)\,. 
\end{equation}
In fact, 
\[\begin{aligned}
L^2(\R^2;dx)&=\bigoplus_{m\in\Z} \Pi_m\big(L^2(\R^2;dx)\big)\simeq \bigoplus_{m\in\Z} L^2(\R_+,rdr)\,,\\
L&=\bigoplus_{m\in\Z} L\Pi_m\simeq\bigoplus_{m\in\Z}H_{m,0}\,,\end{aligned}\]
where
\[H_{m,0}:=  \pi_mL\pi_m^* =-\frac{\partial^2}{\partial r^2}-\frac1{r}\frac{\partial}{\partial r} +\frac{r^2}4+\frac{m^2}{r^2}-m \]
is the self-adjoint operator in $L^2(\R_+,rdr)$ associated with   the quadratic form 
\[q_{m,0}(u)=\int_{\R_+}\left(|u'(r)|^2+\frac{1}{r^2}\Big(m-\frac{r^2}{2} \Big)^2|u|^2\right)rdr\,.\]
Then we get
\[\sigma(L)=\overline{\bigcup_{m\in\Z}\sigma(H_{m,0})}\,.\]
For each Landau level $\Lambda_n$,  we introduce the  set
\[\mathcal  J_n=\{m\in\Z~:~\Lambda_n\in  \sigma(H_{m,0})\}\,.\]
 For a given $m\in\mathcal J_n$, $\Lambda_n$ is a simple  eigenvalue of $H_{m,0}$.  However, since  $\Lambda_n$ is an eigenvalue of $L$ with infinite multiplicity, we deduce    that $\mathcal  J_n$  is infinite. Note that $0\in\mathcal J_1$ and by the min-max principle, $\mathcal J_1\subset[0,+\infty)$.

\subsection{The magnetic harmonic oscillator}\label{sec:m-h-osc}

Consider the case where  $\mathfrak v_0(x)=\mu|x|^2$, where $\mu$  is a  positive constant. The single well operator in \eqref{eq:1well} becomes
\[\mathcal L_h^{\rm sw}= (hD-\Ab)^2+\mu |x|^2\,. \]
After  rescaling\footnote{We do the change of  variable $y=h^{-1/2}x$.} we get
\[  \sigma(\mathcal L_h^{\rm sw})=h\sigma(L_\mu^{\rm  mag}) \]
where
\begin{equation}\label{eq:m-h-osc}
 L_\mu^{\rm  mag}=  (D-\Ab)^2+\mu |x|^2\,.
\end{equation}
We decompose the operator $L_\mu^{\rm  mag}$ via the  orthogonal projections on the Fourier modes as  follows
\[ L_\mu^{\rm  mag}\simeq\bigoplus_{m\in\Z}  H_{m,\mu}\]
where
\[
H_{m,\mu}:= \pi_m L_\mu^{\rm  mag}\pi_m^*
= -\frac{\partial^2}{\partial r^2}-\frac1{r}\frac{\partial}{\partial r}+\Big(\frac{1}{4}+\mu\Big)r^2+\frac{m^2}{r^2}-m\,.
\]

The min-max principle yields
\begin{equation}\label{eq:m-h-osc-}
\forall\,m<0,\quad \lambda_1(H_{m,\mu})> \inf_{u\not=0}\frac{\big\langle(-\Delta +\left(\frac{1}{4}+\mu\right)|x|^2)u,u\big\rangle_{L^2(\R^2)}}{\|u\|_{L^2(\R^2)}}=2\sqrt{\frac14+\mu}\,.
\end{equation}
Moreover, the rescaling  $r\mapsto (1+4\mu)^{1/4}r$ yields  the reduction to the Landau Hamiltonian,
\begin{equation}\label{eq:m-h-osca}
\begin{aligned}
 H_{m,\mu}&=\sqrt{1+4\mu}\left( -\frac{\partial^2}{\partial r^2}-\frac1{r}\frac{\partial}{\partial r} +\frac{r^2}4+\frac{m^2}{r^2}-m\right)+\left(\sqrt{1+4\mu}-1\right)m\\
&=\sqrt{1+4\mu}\,H_{m,0}+\left(\sqrt{1+4\mu}-1\right)m\,.
\end{aligned}
\end{equation}
Consequently,  we infer from \eqref{eq:m-h-osc-} and \eqref{eq:m-h-osc} 
\[ \inf_{m\in\Z}\lambda_1(H_m)=\lambda_1(H_0)=\sqrt{1+4\mu}\,,\quad \inf_{\substack{m\in\Z\\m\not=0}}\lambda_1(H_m)>\sqrt{1+4\mu}\,.\]
This  implies  that
\begin{equation}\label{eq:m-h-osc-ev}
\lambda_1(L_\mu^{\rm mag})=\sqrt{1+4\mu}
\end{equation}
is a simple eigenvalue  and that its (normalized)  associated eigenfunction is radial:
\begin{equation}\label{eq:m-h-osc-ef}
\phi_\mu^{\rm  mag}(x) = \pi^{-1/2}(1+4\mu)^{1/4}\exp\left(-\frac{\sqrt{1+4\mu}\,}{2} |x|^2\right)\,. 
\end{equation}

\subsection{Eigenvalue asymptotics  and radial ground states}

Assuming that  the potential function $\mathfrak v_0$ satisfies \eqref{eq:v0}, we have an accurate description of the spectrum of the operator $\mathcal L_h^{\rm sw}$ introduced in \eqref{eq:1well},  which will provide  an example where the hypotheses imposed by Fefferman--Shapiro--Weinstein in \cite{FSW} hold  (see their Assumption 1.4).

In the sequel we use the notation $\big(\lambda_j(\mathcal P)\big)_{j\in\N}$ for the sequence of min-max eigenvalues of a given self-adjoint operator  $\mathcal P$.

\begin{proposition}\label{prop:m-harm-app}
For every fixed $j\in\N$,  the $j$'th eigenvalue of $\mathcal  L_h^{\rm sw}$ satisfies,
\[\lambda_j(\mathcal L_h^{\rm sw})= v_0^{\min}+h\,  \lambda_j(L_\mu^{\rm mag})+\mathcal O (h^{3/2}) \quad(h\to0_+)\,, \]
where $L_\mu^{\rm mag}$ is the operator introduced in  \eqref{eq:m-h-osc}, with $\mu=\frac{v''_0(0)}{2}$. 

Moreover, the lowest eigenvalue of $\mathcal L_h^{\rm sw}$  is simple with a  radial ground state.
\end{proposition}
 As a consequence of  Proposition~\ref{prop:m-harm-app} and the  construction of accurate quasi-modes, the  lowest eigenvalue  $\lambda_1(\mathcal L_h^{\rm sw})$  can be expanded to any order in  powers of $h$ (see Proposition~\ref{prop:WKB} below).
\begin{proof}[Proof of Proposition~\ref{prop:m-harm-app}]
 Except the last statement, the proof is standard (see \cite{HeSjPise} for the magnetic case) and corresponds to the so-called Harmonic approximation in the case of a non degenerate well (see \cite{CFKS, S} and \cite[Ch.~7]{FH-b}).
We write
\[\mathfrak  v_0(x)=\mathfrak  v_0^{\rm app}(x) +\mathcal  O(|x|^3)\quad(x\to 0)\]
where
\[\mathfrak v_0^{\rm app}(x):=v_0^{\min}+\mu|x|^2\,.\]
For any $C>0$, the spectrum below $v_0^{\min} + C h$ of $\mathcal L_h^{\rm sw}$ is effectively given (modulo $\mathcal O (h^{3/2})$) by that  of
\[ (hD-\Ab)^2+\mathfrak  v_0^{\rm app}(x)=v_0^{\min}+h L_\mu^{\rm mag} \,,\]
so  we are reduced to the operator analyzed in Sec.~\ref{sec:m-h-osc}, thereby getting the asymptotics displayed in Proposition~\ref{prop:m-harm-app}. 

To prove the last statement, we consider a normalized ground state $\psi_h$  of $\mathcal L_h^{\rm sw}$. After rescaling, we obtain  from $\psi_h$ the following normalized function
\[\mathfrak u_h(x):= h^{-1/2}\psi_h(h^{1/2}x)\,. \]  
 Moreover, the operator $\mathcal L_{h}^{\rm sw}$  can be fibered as $\mathcal L_h^{\rm sw}\simeq  h\bigoplus\limits_{m\in\Z} \mathcal  L_{h,m}$, where
 \[  \mathcal L_{h,m}:=\pi_m\mathcal L_h^{\rm sw}\pi_m^*=  -\frac{\partial^2}{\partial r^2}-\frac1r\frac{\partial}{\partial r}+ h^{-1}v_0(h^{1/2}r)+\frac{r^2}{4}+\frac{m}{r^2}-m\,.\]
For $h$ sufficiently  small,  the ground state energy of $\mathcal L_{h}^{\rm sw}$  is simple, so  there exists a unique $m_*\in\Z$  such  that $\mathfrak  u_h=\Pi_{m_*}\mathfrak  u_h$, where $\Pi_{m_*}$ is the projection introduced in \eqref{eq:proj-Fm}.  Note that $m_*$  could depend on $h$ but we skip the reference to $h$ to simplify the presentation.

 The theory of harmonic approximation yields that  the ground state $\mathfrak  u_h$ is close to the normalized radial ground state $\phi_\mu^{\rm mag}$   of the operator  $L_\mu^{\rm mag}$ introduced in \eqref{eq:m-h-osc-ef}.  
  In fact,    we have a spectral gap
\[\delta(\mu):=\lambda_2(L_\mu^{\rm mag})-\lambda_1(L_\mu^{\rm mag})>0\]
and
\[ \|\mathfrak  u_h-\phi_\mu^{\rm mag}\|_{L^2(\R^2)}=\mathcal O(h^{1/2} )\quad(h\to0_+)\,. \]
Now we write  the decomposition
\[ \begin{aligned}
\|\mathfrak  u_h-\phi_\mu^{\rm mag}\|^2_{L^2(\R^2)}&=\sum_{m\in\Z}\|\Pi_m (\mathfrak u_h-\phi_\mu^{\rm  mag})\|_{L^2(\R^2)}^2\\
&=\|\Pi_0 \mathfrak u_h-\phi_\mu^{\rm  mag}\|^2_{L^2(\R^2)}+\sum_{\substack{m\in\Z\\m\not=0}}
\|\Pi_m\mathfrak u_h\|_{L^2(\R^2)}^2\,,
\end{aligned}\]
where we used  that $\Pi_0\phi_\mu^{\rm mag}=\phi_\mu^{\rm mag}$ and, for $m\not=0$, $\Pi_m\phi_\mu^{\rm mag}=0$,  since the function $\phi_{\mu}^{\rm  mag}$  is radial.

Consequently, 
\begin{equation*}
\|\Pi_0\mathfrak  u_h-\phi_\mu^{\rm  mag}\|^2_{L^2(\R^2)} =\mathcal O (h^{1/2})\,,
\end{equation*}
\begin{equation*}
\sum_{\substack{m\in\Z\\m\not=0}}
\|\Pi_m\mathfrak u_h\|_{L^2(\R^2)}^2=\mathcal O(h^{1/2}) \quad(h\to0_+)
\end{equation*}
and 
\[ \mathfrak u_h=\Pi_0\mathfrak u_h+\mathcal O (h^{1/2})\,.\]
 This proves that  $m_*=0$ and that the ground state $\mathfrak u_h$ is radial.
\end{proof}

\section{Decay of ground states for the single well potential}\label{sec:WKB}
Again, we recall standard results  but just take advantage of the additional assumption that the one well potential is radial.
\subsection{The Agmon distance}\label{sec:mag-Agmon}

Consider the  radial  potential function  $\mathfrak w $ on $\R^2$
\begin{equation}\label{eq:W}
\mathfrak w (x):=\frac{|x|^2}{4}+\mathfrak v_0(x)= \frac{r^2}4+v_0(r)\quad(r=|x|\,\mbox{ and }\,x\in\R^2)\,.
\end{equation}
We introduce the  smooth  radial function on $\R^2$ associated with the potential $\mathfrak w$,
\begin{subequations}
\begin{equation}\label{eq:dist-W}
\mathfrak d(x)=d(|x|)
\end{equation}
with 
\begin{equation}\label{eq:d(r)}
d(r):=\int_0^{r}\sqrt{\frac{\rho^2}{4}+v_0(\rho)-v_0^{\min}}\,d\rho 
\end{equation}
\end{subequations}
The function $\mathfrak d$ amounts to the (Agmon) distance to the well $\{0\}$, relative to the potential $\mathfrak w$.
\subsection{Agmon estimates}

If $f$ is a radial function, then
\begin{equation}\label{eq:Lh-rad}\mathcal L_h^{\rm sw} f=-h^2\Delta f+\mathfrak w f 
\end{equation}
 Therefore, when restricting the action of $\mathcal L_h^{\rm sw} $ to radial functions, we  consider $\mathfrak w $ as the effective potential.
Hence, we can apply the semi-classical analysis relative to the Schr\"odinger operator without magnetic potential as considered in \cite{HeSj1} or \cite{S} (see \cite{He88} or \cite{DS} for a more pedagogical presentation).
The identity  in \eqref{eq:Lh-rad} and an integration by  parts yield the following result \cite[Thm.~3.1.1]{He88}.

\begin{proposition}\label{prop:He-Thm3.1.1}
For all $R>0$, let $D_R=\{x\in\R^2~:~|x|<R\}$. If $\phi\in C^0(\overline{D_R};\R)$ and $u\in C^2(\overline{D_R};\R)$  are radial functions such that $\phi$ is Lipschitz and $u=0$ on  $\partial  D_R$,  then the following identity holds
\[ \int_{D_R}\Big(h^2|\nabla(e^{\phi/h}u)|^2+(\mathfrak w -|\nabla\phi|^2)\Big)|e^{\phi/h}u|^2 dx=\int_{D_R}e^{2\phi/h} u\,\mathcal L_h^{\rm  sw}u \,dx\,.\]
\end{proposition}

We have the following standard application of Proposition~\ref{prop:He-Thm3.1.1} on the decay of ground  states of the operator $\mathcal L_h^{\rm  sw}$.
\begin{proposition}\label{prop:He-Prop3.3.1}
For all $\delta\in(0,1)$, there exist $a(\delta),C_\delta,h_0>0$ such that $\lim\limits_{\delta\to0_+}a(\delta)=0$ and, if $\mathfrak u_h$ is a ground state of $\mathcal L_h^{\rm sw}$ and  $h\in(0,h_0)$, then we have,
\[ \left\|\nabla\left(e^{(1-\delta)\mathfrak d(x)/h}\mathfrak u_h\right)\right\|^2+\left\|e^{(1-\delta)d(x)/h}\mathfrak u_h \right\|^2\leq  C_\delta\, e^{a(\delta)/h}\,\|\mathfrak u_h\|^2\,,\]
where $\mathfrak d$ is  the Agmon distance introduced in \eqref{eq:dist-W}.
\end{proposition}
The estimate in \eqref{prop:He-Prop3.3.1} is not optimal since we work under the assumption in \eqref{eq:v0}. In fact, we can write  estimates with $\delta=0$ as we shall see in Proposition~\ref{prop:WKB-app} later on. 
\subsection{WKB approximation}  
For all $S>0$, we introduce the  set
\begin{equation}\label{eq:Om}
B_{\mathfrak d}(S)=\{x\in\R^2~:~\mathfrak d(x)<S\}\,,
\end{equation}
where  $\mathfrak d$ is  the Agmon distance  to $0$ introduced in \eqref{eq:dist-W}.  Since $\mathfrak d$ is  monotone increasing  with respect to $|x|$, 
we have
\begin{equation}\label{eq:B=D}
B_{\mathfrak d}(S)=D(0,R_S):=\{x\in\R^2~:~|x|<R_S\}
\end{equation}
where $R_S$ is the unique   solution of $d(R)=S$. Clearly, $R_S$ is monotone increasing with respect to $S$.

We can then perform the WKB construction outlined in the following proposition.
\begin{proposition}[cf. Prop.~4.4.3 in \cite{He88}]\label{prop:WKB}
There exist $N_0\geq 1$ and  two  sequences $(E_k)_{k\geq 0}\subset\R$  and $(\mathfrak a_k)_{k\geq 0}\subset C^\infty(\R^2)$ such that, for all  $N\geq 1$ and $S>0$,
\[   e^{\mathfrak d(x)/h}\Big(\mathcal L_{h}^{\rm sw}-E^N(h)\Big)\vartheta^N=\mathcal O(h^{N-N_0})\quad{\rm on~}B_{\mathfrak d}(S)\,,\]
where
\begin{align*}
 E^N(h)&=\sum_{k=0}^N E_kh^{k}\,,\quad
E_0=v_0^{\min}, \quad E_1=\sqrt{1+2v''_0(0)}\\
\vartheta^N(x)&=h^{-1/2}\left(\sum_{k=0}^N \mathfrak a_k(x)h^k\right)e^{-\mathfrak d(x)/h},\quad 
\mathfrak a_0(0)=\frac12\sqrt{\frac{1+2v_0''(0)}{\pi}} \,.
\end{align*}
Moreover $\mathfrak  a_0(x)>0$ and  for every $k$,  the  function $\mathfrak a_k$ is radial.
\end{proposition}
\begin{remark}\label{rem:a0}
 The function $\mathfrak a_0$ satisfies the transport equation
  $$2\nabla \mathfrak d\cdot\nabla\mathfrak  a_0+(\Delta \mathfrak d-E_1)\mathfrak a_0=0\,.$$ Since $\mathfrak d$ and $\mathfrak a_0$ are radial, we get
\[\mathfrak a_0(x)=a_0(|x|):=\frac12\sqrt{\frac{1+2v_0''(0)}{\pi}} \exp\left(-\int_0^{|x|}f(\rho)d\rho \right)\]
where
\[f(\rho)=\frac14\frac{u'(\rho)}{u(\rho)}+\frac1{2\rho}-\frac{E_1}{2\sqrt{u(\rho)}}\]
and 
\[u(\rho)=\frac{\rho^2}{4}+v_0(\rho)-v_0^{\min}\,.\]
\end{remark}
\begin{proposition}[cf. Theorem~4.4.4 in \cite{He88}]\label{prop:WKB-app} 
There  exists  $N_0\geq 1$, and  for all $h\in(0,1]$, there exists a   ground state $\mathfrak u_h$   of $\mathcal L_h^{\rm sw}$ such that
\[ \|\mathfrak u_h\|_{L^2(\R^2)}=1\,,\]
and  if $\Omega$ is an open bounded set in $\mathbb R^2$, then  for any $N$ the following holds
\[ \left\| e^{\mathfrak d(x)/h}( \mathfrak u_h-\vartheta^N ) \right\|_{H^2(\Omega)}= \mathcal O(h^{N-N_0})\,.\] 
\end{proposition}

\begin{proof}[Proof of Theorem~\ref{thm:main0}]
The first item in Theorem~\ref{thm:main0} follows from Proposition~\ref{prop:m-harm-app}. Consider the normalized ground state $\mathfrak u_h$ of  $\mathcal L_h^{\rm sw}$ in Proposition~\ref{prop:WKB-app}. By Proposition~\ref{prop:m-harm-app}, $\mathfrak u_h$ is radial. By the Sobolev embedding theorem and  Propositions~\ref{prop:WKB-app}
and \ref{prop:WKB}, we have for $\Omega=D(0,R)$ and  $R>0$,
\[\left\| e^{\mathfrak d(x)/h}( \mathfrak u_h-h^{-1/2}\mathfrak a_0 ) \right\|_{L^\infty(\Omega)} =\mathcal O(h^{1/2})\]
thereby proving that $\mathfrak u_h$  is positive, since $\mathfrak a_0$ is. This proves the second and third items in Theorem~\ref{thm:main0}. 
\end{proof}
\section{About  the measure of the tunneling}\label{sec:comp}
Let us inspect more closely the constants  $S_0$ and $S_a$   appearing  in  the tunneling estimates  in \eqref{eq:main}.   Recall that
\begin{equation}\label{eq:S0}
S_0=\int_0^L\sqrt{\frac{\rho^2}{4}+v_0(\rho)-v_0^{\min}}\,d\rho\,.
\end{equation} 
Since $\mathfrak v_0(x)=v_0(|x|)$ vanishes outside $D(0,a)$, we can rewrite the expression of $S_a$ 
in \eqref{eq:Sa} as follows
\begin{equation}\label{eq:Sa-*}
S_a:=2 \int_0^{a} \sqrt{\frac{\rho^2}{4}+v_0(\rho)-v_0^{\min}}\,d\rho  + \int_a^{L-a} \sqrt{\frac{\rho^2}{4} -v_0^{\min}}\,d\rho\,.
\end{equation}
We now prove  {a variational characterization of  $S_0$ and $S_a$ involving the function
$d$ in \eqref{eq:d(r)}.
\begin{proposition}\label{prop:S0-Sa}
We have
 \begin{equation}\label{eq:S0-min}
S_0=\inf_{0<u<a}\left(d(u)+d(L+u)\right)
\end{equation}
and if $v_0<\frac{L(L-2a)}{4}$,
\begin{equation}\label{eq:Sa-min}
S_a=\inf_{0<u<a}\left(d(u)+d(L-u)\right)\,.
\end{equation}
\end{proposition}}
\begin{proof}~\\
{\bf Proof of \eqref{eq:S0-min}.\\}
The function \[(0,a)\ni  u \mapsto \psi_*(u)=\int_0^u\sqrt{\frac{\rho^2}{4}+v_0(\rho)-v_0^{\min}}\,d\rho+\int_0^{L+u}\sqrt{\frac{\rho^2}{4}+v_0(\rho)-v_0^{\min}}\,d\rho\]
satisfies (for $0<u<a$),
\[ \psi_*'(u)=\sqrt{\frac{u^2}{4}+v_0(u)-v_0^{\min}}+\sqrt{\frac{(L+u)^2}{4}-v_0^{\min}}>0\,.\]
Hence, it is monotone increasing and
\[\min_{0<u<a}\psi_*(u)=\psi_*(0)=S_0\]
where $S_0$ is introduced in \eqref{eq:S0}.\medskip\\
{\bf Proof of  \eqref{eq:Sa-min}.\\}
Consider the function  
\begin{equation}\label{eq:phi*}
(0,a) \ni u \mapsto \varphi_*(u)=\int_0^u  \sqrt{\frac{\rho^2}{4}+v_0(\rho)-v_0^{\min}}\,d\rho+\int_0^{L-u}\sqrt{\frac{\rho^2}{4}+v_0(\rho)-v_0^{\min}}\,d\rho\,.\end{equation}
Notice  that, for $u\in(0,a)$ and $a<\frac{L}2$, we have $a<L-a<L-u$ and
\[\varphi_*'(u)=\sqrt{\frac{u^2}{4}+v_0(u)-v_0^{\min}}-\sqrt{\frac{(L-u)^2}{4}-v_0^{\min}}\,,\]
with 
\[\varphi_*'(0)<0,\quad \varphi_*'(a)<0\]
and
\[\varphi_*'(u)=0{\rm~iff}~v_0(u)=\frac{L(L-2u)}{4}\geq \frac{L(L-2a)}{4}>0\,.\]
Consequently, $\varphi_*'(u)$ can not vanish on $(0,a)$ if we know that
$v_0(u) <\frac{L(L-2a)}{4}$. Under this  assumption, we have
\[\varphi_*'<0~{\rm on~}(0,a)\]
and
\[ \varphi_*(a)<\varphi_*(u)<\varphi_*(0)\,,\]
thereby proving \eqref{eq:Sa-min}. 
\end{proof}

We now consider the constant $\hat  S$ introduced in \eqref{eq:hat-S}.
 Recall  that, by  \eqref{eq:hat-S}, 
\[\hat  S=\min_{r\in[0,a]}g_0(r)\]
where
\[g_0(r)=\frac{Lr}2+d(L-r )+d(r)\,,\] 
and $d(r)$ is introduced in \eqref{eq:d(r)}.
\begin{proposition}\label{prop:hat-S} 
 We have
\begin{equation}\label{eq:S0-Sa-hat-S}
S_a< \hat S<\min\left( S_0,S_a+\frac{La}{2}\right)\,,
\end{equation}
and
if $g_0(r_0)=\hat S$, then $0<r_0<a$.
Moreover, if $v_0'\geq -\frac{L}4$,  then $r_0$ is unique.
\end{proposition}
\begin{proof}
We have, for $0\leq r\leq  a$,
\[ g_0'(r)=\frac{L}2-\sqrt{\frac{(L-r)^2}{4}-v_0^{\min} }+\sqrt{ \frac{r^2}{4}+v_0(r)-v_0^{\min}}\,.\]
We observe that
\[\begin{aligned}
g_0'(0)&=\frac{L}2-\sqrt{\frac{L^2}{4}-v_0^{\min} }<0\,,\\
g'_0(a)&=\frac{L}{2}-\sqrt{\frac{L(L-2a)}{4}+\frac{a^2}{4}-v_0^{\min} }+\sqrt{ \frac{a^2}{4}-v_0^{\min}}\\
&\geq \frac{L}{2}-\sqrt{\frac{L(L-2a)}{4}}>0\,.
\end{aligned}\]
So $g_0$ has a minimum $r_0$ and  $r_0\in(0,a)$. Consequently,
\[\hat S=g_0(r_0)<g_*(0)=S_0\,,\quad \hat S<g_0(a)=S_a+\frac{La}{2}\,,\]
and, by \eqref{eq:Sa-min},
\[\hat S=g_0(r_0)\geq S_a+\frac{Lr_0}{2}>S_a\,.\]
 Finally, we observe that
\[g_0''(r)=\frac{1}{2}\frac{\frac{L-r}{2}}{\sqrt{\frac{(L-r)^2}{4}-v_0^{\min}}}+\frac12\frac{\frac{r}{2}+2v_0'(r)}{\sqrt{\frac{r^2}{4}+v_0(r)-v_0^{\min}}} \,,\]
and if furthermore $g'_0(r)=0$, we have
\[ 
g_0''(r)=\frac{\left(\frac{L}{2}+2v_0'(r)\right)\sqrt{\frac{r^2}{4}+v_0(r)-v_0^{\min}}+\frac{Lr}{4} }{2\left(\frac{L}{2}+\sqrt{\frac{r^2}{4}+v_0(r)-v_0^{\min}}\right)\sqrt{\frac{r^2}{4}+v_0(r)-v_0^{\min}}}\,,\]
which is positive  if $\frac{L}{2}+2v_0'(r)\geq0\,$.
\end{proof}

 For technical reasons  (see Proposition~\ref{prop:hop*}), we need  to minimize,  with respect to $r\in[0,a]$, the following function
\begin{equation}\label{eq:g(r,ep)}
g(r,\varepsilon)=\frac{(1-\varepsilon)Lr}2+d(\sqrt{(L-r)^2+2\varepsilon Lr}\,)+d(r)\,,
\end{equation} 
  where $\varepsilon\in(0,1]$ is fixed and $d(\cdot)$ is introduced in \eqref{eq:d(r)}.  
 
 For all $\varepsilon\in[0,1]$, we  set 
\[S(\varepsilon):=\inf_{0<r<a}g(r,\varepsilon)\,,\]
and notice that $\hat  S=S(0)$ (see Proposition~\ref{prop:hat-S}).
\begin{proposition}[Optimizing w.r.t. $\varepsilon$]\label{prop:g(r,ep)}
 We have $\hat S=\inf\limits_{\varepsilon\in(0,1]}S(\varepsilon)$ and 
there exists  $\varepsilon_0\in(0,1)$ such that, for all $\varepsilon\in(0,\varepsilon_0)$,  
\[S(\varepsilon) =g(r_\varepsilon,\varepsilon)~{\rm with~}0<r_\varepsilon<a\,.\]
 Furthermore, if $v_0'\geq -\frac{L}4$, $r_\varepsilon$ is unique and satisfies $\lim\limits_{\varepsilon\to0}r_\varepsilon=r_0$, where $r_0$ is introduced in Proposition~\ref{prop:hat-S}
\end{proposition}
\begin{proof}
Notice that
\[\begin{aligned}
\frac{\partial g}{\partial\varepsilon}(r,\varepsilon)&=-\frac{Lr}{2}+\frac{Lr}{\sqrt{(L-r)^2+2\varepsilon Lr}}
\sqrt{\frac{(L-r)^2+2\varepsilon Lr}4-v_0^{\min}}\\
&\geq -\frac{Lr}{2}+\frac{Lr}{\sqrt{(L-r)^2+2\varepsilon Lr}}\sqrt{\frac{(L-r)^2+2\varepsilon Lr}4}\\
&=0\,.\end{aligned}\]
So 
\[\inf_{0<\varepsilon\leq 1}g(r,\varepsilon)=g(r,0)=\frac{Lr}{2}+d(L-r)+d(r)\,.\]
Consequently
\begin{equation}\label{eq:hat-S=inf}
\hat S=\inf_{r\in[0,a]}g(r,0)=\inf_{r\in[0,a]}\left(\inf_{\varepsilon\in(0,1]}g(r,\varepsilon)\right)=\inf_{\varepsilon\in(0,1]} S(\varepsilon)\,.
\end{equation}
We have seen in the proof of Proposition~\ref{prop:hat-S} that
\[ \frac{\partial g}{\partial r}(0,0)<0\quad{\rm and}\quad \frac{\partial g}{\partial r}(a,0)>0\,.\]
By continuity of $\frac{\partial g}{\partial r}(r,\varepsilon)$ w.r.t. $\varepsilon$, we know that, for $\varepsilon$ sufficiently small
\[ \frac{\partial g}{\partial r}(0,\varepsilon)<0\quad{\rm and}\quad \frac{\partial g}{\partial r}(a,\varepsilon)>0\,.\]
This yields that every minimum $r_\varepsilon$ of $g(r,\varepsilon)$ is in $(0,a)$.

To study the behavior of $r_\varepsilon$  as $\varepsilon\to0$, we start by noticing that
\begin{equation}\label{eq:cont-S-ep}
\lim_{\varepsilon\to0}S(\varepsilon)=\hat S\,.
\end{equation}
In fact, by \eqref{eq:hat-S=inf}, for every $\delta>0$, there exists $\varepsilon_\delta$ such that
\[\hat S\leq S(\varepsilon_\delta)\leq \hat S+\delta\]
and by the monotnicity of $g(r,\cdot)$ we get
\[\forall\,\varepsilon\in(0,\varepsilon_\delta),\quad \hat S\leq S(\varepsilon)\leq S(\varepsilon_0)\leq \hat S+\delta\]
which proves \eqref{eq:cont-S-ep}.
 
Now, consider a sequence $(\varepsilon_n)_{n\geq 1}\subset \R_+$ that converges to $0$ such that $\lim\limits_{n\to+\infty}r_{\varepsilon_n}=r_0\in[0,a]$. Then
\[S(\varepsilon_n)=g_{\varepsilon_n}(r_{\varepsilon_n},\varepsilon_n)\to g(r_0,0)\]
thereby, in light of \eqref{eq:cont-S-ep},  $g(r_0,0)=\hat S$. Thus, by Proposition~\ref{prop:hat-S}, we have  $0<r_0<a$. If moreover  $v_0'\geq -\frac{L}4$, then $r_0$  is the unique minimum of $g_0(\cdot):=g(\cdot,0)$, hence $\lim\limits_{\varepsilon\to0}r_\varepsilon=r_0$\,.  By a continuity argument, for $\varepsilon$ small enough,  $\frac{\partial^2 g}{\partial r^2}(r_\varepsilon,\varepsilon)>0$,
hence the uniqueness of the   minimum of $g(\cdot,\varepsilon)$.
\end{proof}
\section{Bounds on the eigenvalue splitting}\label{sec:hopping}
\subsection{The hopping coefficient}
In \cite{FSW}, the following term has been introduced (\emph{the hopping coefficient})
\begin{equation}\label{eq:hat-w}
w_{\ell,r}=\int_{D(0,a)} \mathfrak v_0(x) \mathfrak  u_h(x+z) \mathfrak u_h(x)\,e^{i\frac{Lx_2}{2h}}\,dx\,,
\end{equation}
where  $\mathfrak u_h$ is the positive ground state of the single well operator $\mathcal L_h^{\rm sw}$  (see Theorem~\ref{thm:main0}), and $z=(L,0)$. In the framework of \cite{HeSj1} (see also \cite{He88}),  $w_{\ell,r}$ can be derived    by a reduction to an interaction matrix.

Assuming the condition \eqref{assFSW}, the following holds by \cite[Eq.~(1.10)~\&~(1.12)]{FSW}
\begin{equation}\label{eq:tun-w}
e_2^{\mathfrak v_0}(h)-e_1^{\mathfrak v_0}(h)\underset{h\to0}{\sim} 2|w_{\ell,r}| \,,
\end{equation}
so the magnitude of the tunneling effect is given by $w_{\ell,r}$.

 Estimating the hopping coefficient $w_{\ell,r}$ being interesting on its own \cite{SW}, the  sharp estimates of  $w_{\ell,r}$ which we provide  in Propositions~\ref{prop:hat-w}, \ref{prop:w} and in Sec.~\ref{sec:asy} do not require the assumption $L>4(a+\sqrt{|v_0^{\min}|})$ in \eqref{assFSW}.  It appears that the foregoing condition is only used to vindicate the asymptotics in \eqref{eq:tun-w}, and relaxing it  is definitely an interesting task.

 As pointed  out in \cite{FSW}, the oscillatory complex phase in the expression of $w_{\ell,r}$ is behind the difficulties in dealing with this term. Fortunately,  as observed in \cite{FSW}, it is  possible to rule  out the oscillatory complex phase  by using special functions.   Recall that the ground state $\mathfrak u_h(x)=u_h(|x|)$ is a radial function and $a=a(\mathfrak v_0)$ is introduced in \eqref{eq:def-a}.
\begin{proposition}[Bounds on $ w_{\ell,r}$]\label{prop:hat-w}~
 We  have $ w_{\ell,r} \in\R$, and there exists $c_2>0$ such that, for all $h\in(0,1]$, we have
\[|w_{\ell,r}|\leq c_2\int_0^a |v_0(r)| u_h(L-r)u_h(r)rdr\,.\]
Furthermore, if $v_0\leq  0$,  then  for all $\varepsilon\in(0,1]$, there exists $c_\varepsilon>0$ such  that, for all $h\in(0,1]$, we have
\[ 
|w_{\ell,r}|\geq  c_\varepsilon \int_0^a e^{-\frac{(1-\varepsilon)Lr}{2h}}  |v_0(r)|u_h\big(\sqrt{(L-r)^2+2\varepsilon Lr}\,\big)u_h(r)rdr.\]
\end{proposition} 
 The proof of Proposition~\ref{prop:hat-w} relies on the tricky  representation, in Lemma~\ref{lem:FSW-uh} below, of the function  $u_h$, defining the radial ground state $\mathfrak u_h$.  It is obtained in \cite{FSW}, but with the expansion of $u_h$  in Theorem~\ref{thm:main0}, we can describe the coefficients in a sharper manner than in \cite{FSW}.

With $a=a(\mathfrak v_0)$ in \eqref{eq:def-a}, $d(\cdot)$ in \eqref{eq:d-main0} and $\mathfrak a_0(\cdot)$ in \eqref{eq:uh-main0}, we introduce the two constants
\begin{equation}\label{eq:def-m-f}
\begin{aligned}
F(\mathfrak v_0)&=\frac{a}{4}\sqrt{a^2+4|v_0^{\min}|}+\frac{|v_0^{\min}|}2\ln\frac{\big(\sqrt{a^2+4|v_0^{\min}|}+a\big)^2 }{4|v_0^{\min}|}-d(a)\\
m(\mathfrak v_0)&=\mathfrak a_0(0)\sqrt{\frac{2a|v_0^{\min}|}{\pi}}\frac{\left(a^2+4|v_0^{\min}| \right)^{1/4}}{\sqrt{a^2+4|v_0^{\min}|}+a} \,.
\end{aligned}
\end{equation}
\begin{lemma}\label{lem:FSW-uh}
The function  $u_h$, defining the ground  state $\mathfrak u_h$ in Theorem~\ref{thm:main0}, has the following representation, valid for $\rho\geq a$,
\begin{equation}\label{eq:FSW-uh} 
u_h(\rho)=C_h\exp\left( -\frac{\rho^2}{4h}\right)\int_0^{+\infty} \exp\left( -\frac{\rho^2t}{2h} \right)
t^{\alpha-1}(1+t)^{-\alpha}dt\,,
\end{equation}
where
\begin{equation}\label{eq:def-alpha}
\alpha=\frac1{2h}|v_0^{\min}|-\frac{1}2\left(\sqrt{1+2v''_0(0)}-1 \right) +\mathcal O(h^{1/2})\underset{h\to0}{\sim} \frac1{2h}|v_0^{\rm min}|\,,
\end{equation}
and
\begin{equation}\label{eq:Ch}
C_h \underset{h\to0}{\sim} C_h^{\rm asy}:= \mathfrak m(v_0)  h^{-1} \exp\left(\frac{F(\mathfrak v_0)}{h}  \right)\,.
\end{equation}
\end{lemma}
\begin{proof}
The representation in \eqref{eq:FSW-uh} is obtained in \cite[Eq.~(2.9)]{FSW}, with 
\begin{equation}\label{defalpha}
\alpha=\frac{1}{2}-\frac{1}{2h} e^{\rm sw}(h)\,.
\end{equation}
The asymptotics in \eqref{eq:def-alpha}  now follows from \eqref{eq:gs} in Theorem~\ref{thm:main0}. So we still have to determine the constant $C_h$, by matching \eqref{eq:FSW-uh} with the expansion of $u_h(\rho)$ in Theorem~\ref{thm:main0}. In fact, by  \eqref{eq:uh-main0} and \eqref{eq:d-main0}, we have 
\[ u_h(a)\underset{h\to0}{\sim} \mathfrak a_0(0) h^{-1/2}e^{-\frac{d(a)}{h}  }\,,\]
where $ \mathfrak a_0(0)$ is given in \eqref{eq:defa0}.

On the other hand,  by the method of Laplace approximation \cite[Eq.~(2.12)]{FSW}, the representation in \eqref{eq:FSW-uh} yields
\[ u_h(a)\underset{h\to0}{\sim} C_h\sqrt{\frac{2\pi h}{|v_0^{\min}|(1+2t_*(a))}} (1+t_*(a))e^{-\frac{\eta(a)}{h}}\,, \]

\[t_*(a)=\frac12\left(\sqrt{1+\frac{4}{a^2}|v_0^{\min}| } -1\right)\]
and
\[\eta(a)=\frac14(1+2t_*(a))a^2+\frac{|v_0^{\min}|}{2}\ln\left(1+\frac1{t_*(a)}\right) \,.\] 
 Consequently, we have
\[C_h\underset{h\to0}{\sim} \frac{\sqrt{|v_0^{\min}|(1+2t_*(a))}}{\sqrt{2\pi}(1+t_*(a))} h^{-1} e^{-\frac{d(a)-\eta(a)}{h}}\]
which eventually  yields \eqref{eq:Ch}. 
\end{proof}

\begin{proof}[Proof of Proposition~\ref{prop:hat-w}]
We start by expressing the integral \eqref{eq:hat-w} in polar  coordinates
\begin{equation}\label{eq:w-int}
w_{\ell,r}=\int_0^a  r\,v_0(r) u_h(r) \left(\int_0^{2\pi}K_h(r,\theta)d\theta\right)dr\,, 
\end{equation}
where
\[K_h(r,\theta)= u_h\left(\sqrt{r^2+L^2+2Lr\cos\theta}\right) e^{\frac{iLr\sin\theta}{2h}}\,.\] 
Coming  back to \eqref{eq:w-int}, the integral of $K_h$ with respect to  $\theta$ is computed in \cite[Prop.~5.1]{FSW} by using \eqref{eq:FSW-uh} as follows
\begin{equation}\label{eq:int-Kh} \int_0^{2\pi}K_h(r,\theta)d\theta =
C_h\exp\left( -\frac{r^2+L^2}{4h}\right)\int_0^{+\infty}
G_h(r,t) dt\,,
\end{equation}
where 
\begin{equation}\label{eq:def-Gh}
 G_h(r,t)=
\exp\left( -\frac{(r^2+L^2)t}{2h} \right)
t^{\alpha-1}(1+t)^{-\alpha}I_0\left(\frac{Lr\sqrt{t(t+1)}}{h}\right)
\end{equation}
and $z\mapsto I_0(z):=J_0(iz)$ is the modified Bessel's function of order $0$.\\
The advantage of the  formula in \eqref{eq:int-Kh}
 is the absence of the oscillatory  complex term and moreover, the integrand $G_h$ is a positive function.

 The function $I_0(z)$ has the following asymptotic  for large $z$ \cite[Eq. (9.3.14)]{T},   
\begin{equation}\label{eq:I0-main}
I_0(z)\underset{z\to+\infty}{\sim} \frac{e^z}{\sqrt{2\pi z}}\,.
\end{equation} 
Furthermore, $I_0$ can be represented as a convergent series with positive terms
\[ I_0(z)= \sum_{k=0}^{+\infty} \frac{(\frac 14 z^2)^k}{(k!)^2}\quad(z\in\R)\,.\]
In particular,  $I_0$ grows exponentially as follows 
\begin{equation}\label{eq:I0}
\frac{c_1}{\sqrt{2\pi z}+1}e^z\leq  I_0(z)\leq\frac{c_2}{\sqrt{2\pi z}+1} e^z\leq c_2 e^z \quad (z\in\R_+)\,,
\end{equation}
where $c_1,c_2$ are positive constants. \\   
We introduce
\begin{align*}
 F(r)&=\int_0^{+\infty}
\exp\left( -\frac{(r^2+L^2)t}{2h} \right)
t^{\alpha-1}(1+t)^{-\alpha} \exp\left(\frac{Lr\sqrt{t(t+1)}}{h}\right)dt\\
&=\int_0^{+\infty}
\exp\left( -\frac{(L-r)^2 t}{2h}\right)
t^{\alpha-1}(1+t)^{-\alpha}\exp\left(\frac{Lr}{h}\frac{\sqrt{t}}{\sqrt{t}+\sqrt{t+1}}\right)dt\,.
\end{align*}
Note that, for all $t>0$, we have
\[
 0\leq  \frac{\sqrt{t}}{\sqrt{t}+\sqrt{t+1}}\leq \frac12\]
hence
\begin{equation}\label{eq:F-2}
F(r)\leq  F_2(r)
\end{equation}
where
\[F_2(r)=e^{\frac{Lr}{2h}}\int_0^{+\infty}
\exp\left( -\frac{(L-r)^2 t}{2h}\right)
t^{\alpha-1}(1+t)^{-\alpha}dt\,.\]
Collecting \eqref{eq:int-Kh}, \eqref{eq:I0} and \eqref{eq:F-2}, we get by \eqref{eq:FSW-uh}
\begin{equation}\label{eq:K<uh}
 \int_0^{2\pi}K_h(r,\theta)d\theta \leq   c_2 u_h(L-r) \,.
\end{equation}
Let  us now bound $\int_0^{2\pi}K_h(r,\theta)d\theta$. Given an arbitrary $\varepsilon\in(0,1]$, it results from \eqref{eq:I0} that there exists a constant $c_\varepsilon>0$ such that 
\begin{equation}\label{eq:I0*}
 I_0(z)\geq c_\varepsilon e^{(1-\varepsilon)z}\quad (z\in\R_+) \,.\end{equation}
We now introduce,
\begin{align*}
 F^\varepsilon (r)&=\int_0^{+\infty}
\exp\left( -\frac{(r^2+L^2)t}{2h} \right)
t^{\alpha-1}(1+t)^{-\alpha} \exp\left(\frac{(1-\varepsilon)Lr\sqrt{t(t+1)}}{h}\right)dt\\
&=\int_0^{+\infty}
\exp\left( -\frac{(L-r)^2 t}{2h}\right)
t^{\alpha-1}(1+t)^{-\alpha}\exp\left(\frac{Lr}{h}\big((1-\varepsilon)\sqrt{t(t+1)}-t\big)\right)dt\,.
\end{align*}
Note that, for all $t>0$, we have
\[ (1-\varepsilon)\sqrt{t(t+1)}-t\geq (1-\varepsilon)t-t=-\varepsilon t\]
hence
\begin{equation}\label{eq:F-1}
F^\varepsilon(r)\geq F_1^\varepsilon(r)\end{equation}
where
\[ F_1^\varepsilon(r) =\int_0^{+\infty}
\exp\left( -\frac{\big((L-r)^2+2\varepsilon Lr\big)t}{2h}\right)
t^{\alpha-1}(1+t)^{-\alpha}dt\,.\]
Collecting \eqref{eq:int-Kh}, \eqref{eq:I0*} and \eqref{eq:F-1}, we get by \eqref{eq:FSW-uh}
\[ \int_0^{2\pi}K_h(r,\theta)d\theta\geq  c_\varepsilon e^{-\frac{(1-\varepsilon)Lr}{2h}}\, u_h\big(\sqrt{(L-r)^2+2\varepsilon Lr}\,\big) \,.\]
Recalling  that $0<r<a$, we  get further
\[ \int_0^{2\pi}K_h(r,\theta)d\theta\geq  c_\varepsilon e^{-\frac{(1-\varepsilon)La}{2h}}\,u_h\big(\sqrt{(L-r)^2+2\varepsilon Lr}\,\big) \,.\]
\end{proof}

 \subsection{WKB approximation}

Using \eqref{eq:uh-main0} and  Proposition~\ref{prop:hat-w} (with $\varepsilon=1$), we get,  
\begin{equation}\label{eq:w-N}
c_1w_{\ell,r}^{0,-}+\mathcal O(M_h^-)\leq |w_{\ell,r}|\leq  c_2w_{\ell,r}^{0,+}+\mathcal O(M_h^+)\,, 
\end{equation}
where
\begin{align*}
w_{\ell,r}^{0,+}&=  h^{-1}\int_0^a |v_0(r)| a_0(L-r)a_0(r)\,\exp\left(-\frac{d(r)+d(L-r)}{h}\right)\,r\,dr\,,\\
w_{\ell,r}^{0,-}&=  h^{-1}\int_0^a |v_0(r)| a_0(L+r)a_0(r)\,\exp\left(-\frac{d(r)+d(L+r)}{h}\right)\,r\,dr\,.
\end{align*}
and
\begin{equation}\label{eq:def-M}
\begin{aligned}
M_h^+&=\int_0^a |v_0(r)|\exp\left(-\frac{d(r)+d(L-r)}{h}\right)\,r\,dr\,,\\
M_h^-&=\int_0^a |v_0(r)|\exp\left(-\frac{d(r)+d(L+r)}{h}\right)\,r\,dr\,.
\end{aligned}
\end{equation}
The remainder terms $M_h^\pm$ are easily controlled as follows.
\begin{proposition}\label{prop:M-h}
We have
\[ M_h^+=\mathcal O(e^{-S_a/h}) \quad {\rm and}\quad M_h^-=\mathcal O(e^{-S_0/h})\,.\]
\end{proposition}
\begin{proof}
By \eqref{eq:Sa-min} and \eqref{eq:S0-min},
\[\min_{0\leq r\leq a}\big( d(r)+d(L-r)\big) =S_a~{\rm and~} \min_{0\leq r\leq a}\big( d(r)+d(L+r)\big)=S_0\]
where $S_a$ and $S_0$ are  introduced in  \eqref{eq:Sa} and \eqref{eq:S0} respectively. Hence
\[M_h^+\leq e^{-S_a/h}\int_0^a |v_0(r)|  \,r\,dr \]
and
\[M_h^-\leq e^{-S_0/h}\int_0^a |v_0(r)|   \,r\,dr\,.  \]
\end{proof}
We move now to the control of the leading terms, $w_{\ell,r}^{0,\pm}$, in \eqref{eq:w-N}. 
\begin{proposition}\label{prop:w}
There exist constants $h_0,C>0$ such that, for all $h\in(0,h_0]$, we have
\[ w_{\ell,r}^{0,+}\leq Ch^{-1}e^{-S_a/h}\]
and
\[ w_{\ell,r}^{0,-}\geq \frac{h}{C} e^{-S_0/h}\,.\]
\end{proposition}
\begin{proof}
The bound on $ w_{\ell,r}^{0,+}$ follows in a straightforward manner, as for the bound on $M_h^+$ in Proposition~\ref{prop:M-h}.
Concerning $w_{\ell,r}^{0,-}$, 
the function $\psi_*(r)=d(r)+d(L+r)$ is monotone increasing and
\[ S_0=\psi_*(0)= \min_{0  \leq r\leq \eta}\psi_*(r),\quad \psi_*'(0)=\sqrt{\frac{L^2}{4}-v_0^{\min}}>0\,.\]
By Laplace's approximation, 
\[\begin{aligned}
w_{\ell,r}^{0,-}&\underset{h\to0}{\sim} e^{-S_0/h}\int_0^a |v_0(0)| a_0(L)a_0(0)\exp\left(-\frac{\sqrt{\frac{L^2}{4}-v_0^{\min}}\,r}{h}\right)\,r\,dr \\
&= \frac{h^2}{\frac{L^2}{4}-v_0^{\min}} e^{-S_0/h}\,.
\end{aligned}\]
\end{proof}

Collecting  \eqref{eq:w-N}, Propositions~\ref{prop:M-h} and \ref{prop:w}, we get 
\begin{proposition}\label{prop:hop}
There exist constants $h_0,\tilde C>0$ such that, for all $h\in(0,h_0]$, we have
\[ \frac{h}{\tilde C}\exp\left(-\frac{S_0}{h} \right)\leq |w_{\ell,r}|\leq \tilde Ch^{-1} \exp\left(-\frac{S_a}{h}\right)\,.\]
\end{proposition}

By inserting the estimates in Proposition~\ref{prop:hop} into \eqref{eq:tun-w}, we finish the proof of \eqref{eq:main} in Theorem~\ref{thm:main}. So we still  have to prove \eqref{eq:main*} which follows from the following proposition.

\begin{proposition}\label{prop:hop*}
If $v_0<0$ on $D(0,a)$, then
\begin{equation}\label{eq:w-v0<0}
\liminf_{h\to0} \big(h\ln|w_{\ell,r}|\big)\geq   -\hat S \,,
\end{equation}
where $\hat S$ is introduced in \eqref{eq:hat-S}.
 \end{proposition}
\begin{proof}
We use Proposition~\ref{prop:hat-w} with $0<\varepsilon<1$ and  we replace $u_h$ by its WKB approximation using \eqref{eq:uh-main0}. Eventually we get
\begin{equation}\label{eq:w-ep}
|w_{\ell,r}|\geq c_\varepsilon w_{\ell,r}^\varepsilon-C_\varepsilon M_h^\varepsilon\,,\end{equation}
where $c_\varepsilon,C_\varepsilon>0$ are independent of $h$,
\[w_{\ell,r}^\varepsilon=h^{-1}\int_0^a  |v_0(r)|a_0\big(\sqrt{(L-r)^2+2\varepsilon Lr}\,\big)a_0(r)\exp\left(-\frac{g(r,\varepsilon)}{h}\right)\,r\,dr\,,\]
\[M_h^\varepsilon= \int_0^a  |v_0(r)| \exp\left(-\frac{g(r,\varepsilon)}{h}\right)\,r\,dr\]
and $g(r,\varepsilon)$ is  introduced in \eqref{eq:g(r,ep)}.  We can  rewrite \eqref{eq:w-ep} as follows
\begin{equation}\label{eq:w-ep*}
h|w_{\ell,r}|\geq c_\varepsilon\int_0^a  |v_0(r)|\big(a_0\big(\sqrt{(L-r)^2+2\varepsilon Lr}\,\big)a_0(r)-C^*_\varepsilon h\big)\exp\left(-\frac{g(r,\varepsilon)}{h}\right)\,r\,dr
\end{equation}
where $C_\varepsilon^*=C_\varepsilon/c_\varepsilon$.

Now we assume that  $0<\varepsilon<\varepsilon_0$, where  $\varepsilon_0$ is the constant in Proposition~\ref{prop:g(r,ep)}. Pick $r_\varepsilon\in(0,a)$ such that 
\[S(\varepsilon):=\inf_{0<r<a}g(r,\varepsilon)=g(r_\varepsilon,\varepsilon)\,.\]
We choose $\eta_0\in(0,a)$ sufficiently small such that, for all $\eta\in(0,\eta_0)$ and  $r\in I_\eta:=(-\eta+r_\varepsilon,r_\varepsilon+\eta)\subset(0,a)$, we have
\[g(r,\varepsilon)\leq S(\varepsilon) +m_{\varepsilon} \eta\]
where $m_\varepsilon$ is independent of $\eta$. By the same considerations, we have (after choosing $\eta_0,h_0$ small and taking $\eta\in(0,\eta_0)$ and $h\in(0,h_0)$)
\[a_0\big(\sqrt{(L-r)^2+2\varepsilon Lr}\,\big)a_0(r)-C^*_\varepsilon h>m_\varepsilon^*>0\]
where $m_\varepsilon$ is independent of $h$ and $\eta$.

Now we write,
\[ |h\,w_{\ell,r}^\varepsilon|\geq   K_\varepsilon(\eta) \exp\left(-\frac{S(\varepsilon)+m_\varepsilon\eta}{h}\right) \,,
\]
where 
\[K_\varepsilon(\eta)=c_\varepsilon m_\varepsilon^*\int_{I_\eta}|v_0(r)|\,r\,dr\,.\]
The hypothesis $v_0<0$ in $D(0,a)$ ensures that $K_\varepsilon(\eta)>0$. Consequently
\[\liminf_{h\to0}\big(h\ln|h w_{\ell,r}|\big)\geq -S(\varepsilon)+m_\varepsilon\eta\]
Sending  $\eta$ to $0$ then $\varepsilon$ to $0$, we get 
\[\liminf_{h\to0}\big(h\ln|h  \,w_{\ell,r}|\big)\geq -\hat S\,,\]
where we used that $\hat S=S(0)$ by Proposition~\ref{prop:hat-S}.
\end{proof}

\begin{remark}\label{rem:extension}  
The upper bound in Proposition~\ref{prop:hop} continues to  hold if we relax the assumption on the sign of $\mathfrak v_0$ and assume instead  that 
\[ v_0< \frac{L(L-2a)}4\]
which  ensures the validity  of \eqref{eq:Sa-min}. 
\end{remark}

\begin{remark}[Reduction to an interaction matrix]\label{rem:extension*} 
In \cite{FSW}, the asymptotics  \eqref{eq:tun-w} holds for  $v_0\leq 0$ and $L>4(a+\sqrt{|v_0^{\min}|})$. We can derive \eqref{eq:tun-w} by using  the approach of reduction to  an interaction matrix like in \cite{HeSj1} (see \cite{He88} or \cite{DS}).    That is the subject of the work \cite{HKS}  (see also \cite{FMR}) which yields
\begin{equation}\label{eq:tun-w*}
e_2^{\mathfrak v_0}(h)-e_1^{\mathfrak v_0}(h)=2|w_{\ell,r}|+\mathcal O(e^{-2\tilde S/h})\,,
\end{equation}
where $\tilde S$ is a  constant satisfying
$0<\tilde S<2\hat S$ and \(\hat S\) is introduced in \eqref{eq:hat-S}.
\end{remark}

\section{The tunneling asymptotics}\label{sec:asy}
 
 \subsection{Towards the proof  of the main theorem.}

Recall the  constant  $a=a(\mathfrak v_0)$ in \eqref{eq:def-a} and let  us assume that $\mathfrak  v_0<0$ in $D(0,a)$ (i.e. $v_0(r)<0$ for $0\leq r<a$). With Proposition~\ref{prop:hop*} in hand, we can compute the leading term in the hopping coefficient $w_{\ell,r}$ introduced in \eqref{eq:hat-w}. This will follow through a sequence of reductions justified in the  following lemmas.

We  will make use of the following consequence of Proposition~\ref{prop:hop*}. Recall the constant $\hat S$ in \eqref{eq:hat-S}.  Then, by Proposition~\ref{prop:hop*},  we have
\begin{equation}\label{eq:lb-w}
\forall\,p\in\R,\,\forall\,A>\hat S,~ h^{p}e^{-A/h} \underset{h\to0}{=}o(|w_{\ell,r}|)\,.
\end{equation}
Moreover, for all $\eta\in(0,a)$, let us introduce 
\[\mathcal W_1(\eta):= C_h\int_\eta^a r|v_0(r)|u_h(r)\exp\left( -\frac{r^2+L^2}{4h}\right)\left(\int_{\eta}^{+\infty} G_h(r,t)\,dt\right)dr \,,\]
where $C_h$  is introduced in Lemma~\ref{lem:FSW-uh} and $G_h$ is introduced in \eqref{eq:def-Gh}. The next lemma establishes that $|w_{\ell,r}|$ can be approximated  by $\mathcal W_1(\eta)$ uniformly with respect to $\eta\in(0,\eta_0]$, where $\eta_0$ is a sufficiently small constant.
\begin{lemma}\label{lem:exp-w}
 There exist constants $\eta_0\in(0,a)$, $A(\eta_0)>\hat S$ and  $h_0>0$ such that, for all $\eta\in(0,\eta_0]$ and $h\in(0,h_0]$, we have
\[ \big| |w_{\ell,r}|-\mathcal W_1(\eta) \big|\leq e^{-A(\eta_0)/h}\,.\]
In particular, for all $\eta\in(0,\eta_0]$,
\[ |w_{\ell,r}|\underset{h\to0}{\sim} \mathcal W_1(\eta)\,.\]
\end{lemma}
\begin{proof}
Consider $\eta\in(0,\eta_0)$ where $0<\eta_0<a$. By  \eqref{eq:w-int} and \eqref{eq:int-Kh}  we write
\[|w_{\ell,r}|=\mathcal W_1(\eta)+\mathcal R_1(\eta) +\mathcal R_2(\eta)\]
where
\begin{align}
\mathcal R_1(\eta)&=\int_0^\eta r|v_0(r)|u_h(r)\left(\int_{0}^{2\pi} K_h(r,\theta)\,d\theta\right)dr\,, \label{eq:def-R1}\\
\mathcal R_2(\eta)&=C_h\int_\eta^a r|v_0(r)|u_h(r)\exp\left( -\frac{r^2+L^2}{4h}\right)\left(\int_{0}^{\eta} G_h(r,t)\,dt\right)dr\,. 
\end{align}
Using \eqref{eq:K<uh}, we have
\[ 0\leq \mathcal R_1(\eta)\leq \mathcal R_1(\eta_0)\leq c_2\int_0^{\eta_0} r|v_0(r)|u_h(r)u_h(L-r)\,dr\,. \]
Arguing as in Proposition~\ref{prop:w}, and using the monotonicity of the function $\varphi_*$  in \eqref{eq:phi*},  we have
\[\int_{0}^{\eta_0} r|v_0(r)|u_h(r)u_h(L-r)\,dr=\mathcal O(h^{-1} e^{-S_{\eta_0}/h}) \]
where $S_{\eta_0}=\varphi_*(\eta_0)$ 
 and depends continuously on $\eta_0$ so that $\lim\limits_{\eta_0\to0}S_{\eta_0}=S_0$; here $S_0$ is introduced in \eqref{eq:S0}. By Proposition~\ref{prop:hat-S}, we choose $\eta_0$ sufficiently small such that 
$S_{\eta_0} >\hat  S$. This proves that
$\mathcal R_1(\eta_0)=\mathcal O(h^{-1} e^{-S_{\eta_0}/h})\underset{h\to0}{=}o(|w_{\ell,r}|)$, by   \eqref{eq:lb-w}.

Let us now prove  that 
$\mathcal R_2(\eta)\underset{h\to0}{=}o(|w_{\ell,r}|)$.  Using  \eqref{eq:def-Gh} and \eqref{eq:I0}, we write for $0<t<\eta<\eta_0<a$,
\[G_h(r,t)= \mathcal O \big(e^{c_0\sqrt{\eta}/h}\big) C_h\exp\left( -\frac{(r^2+L^2)t}{2h} \right)
t^{\alpha-1}(1+t)^{-\alpha}  \]
where $c_0=\sqrt{a}+1$. Inserting this into  the expression of $\mathcal R_2(\eta)$, we get by \eqref{eq:FSW-uh},
\[\mathcal R_2(\eta)= \mathcal O \big(e^{c_0\sqrt{\eta}/h}\,\big) \int_\eta^a r|v_0(r)|u_h(r)u_h(\sqrt{L^2+r^2}) dr\,.\]
Arguing as in Proposition~\ref{prop:w} and observing that
\[ \min_{\eta<r<a} d(r)+d(\sqrt{L^2+r^2})=d(\eta)+d(\sqrt{L^2+\eta^2})>S_0\]
we get
\[\mathcal R_2(\eta)= \mathcal O \big(h^{-1} e^{c_0\sqrt{\eta}/h}e^{-S_0/h}\,\big)\,.\]
Finally, by  Proposition~\ref{prop:hat-S}, we choose $\eta_0$ sufficiently small so that, 
$S_0-c_0\sqrt{\eta_0}>\hat S$
 and we  conclude by Proposition~\ref{prop:w} that $\mathcal R_2(\eta)\underset{h\to0}{=}o(|w_{\ell,r}|)$, for  all $\eta\in(0,\eta_0)$.  
 
 Looking closely at the foregoing bounds  on $\mathcal  R_1(\eta)$ and $\mathcal R_2(\eta)$, we have in fact proved the following. If we select $A$ such that
 \[ \hat S<A<\min(S_{\eta_0}, S_0-c_0\sqrt{\eta_0} )\]
 then there exists  $h_0>0$ such that,  for all $\eta\in(0,\eta_0]$,
 \[  0\leq \mathcal R_1(\eta)+\mathcal R_2(\eta)\leq e^{-A/h}\,.\]
\end{proof}
 For all $\eta\in(0,a)$, we introduce
\[\mathcal W_2(\eta):= C_h^{\rm asy}\int_\eta^a r|v_0(r)|a_0(r)\exp\left( -\frac{d(r)}{h}-\frac{r^2+L^2}{4h}\right)\left(\int_{\eta}^{+\infty} G_h(r,t)\,dt\right)dr \,,\]
where $C_h^{\rm asy}$ is introduced in \eqref{eq:Ch}, $G_h$ is introduced in \eqref{eq:def-Gh} and $a_0$ is introduced in Proposition~\ref{prop:WKB}. We will prove that $\mathcal W_1(\eta)$ in Lemma~\ref{lem:exp-w} can be approximated by $h^{-1/2}\mathcal W_2(\eta)$, uniformly with  respect to  $\eta\in(0,\eta_0]$.
\begin{lemma}\label{lem:exp-w*}
Let  $\eta_0$ be as introduced  in Lemma~\ref{lem:exp-w}. There exist constants $h_0,M>0$ such that, for all  $\eta\in(0,\eta_0]$ and $h\in(0,h_0]$, we have
\[\big|\mathcal W_1(\eta)  -h^{-1/2}\mathcal W_2(\eta)\big|\leq M h^{1/2} \mathcal  W_2(\eta)+\left|1-\frac{C_h^{\rm asy}}{C_h}\right|\mathcal W_1(\eta)\]
where $\mathcal W_1(\eta)$ is as in  Lemma~\ref{lem:exp-w}.
\end{lemma}
\begin{proof}
By Lemma~\ref{lem:exp-w}, it  suffices to prove that $\mathcal W_1(\eta)\underset{h\to0}{\sim} h^{-1/2}\mathcal W_2(\eta)$  (and estimate the remainder terms).  We write
\[ g_\eta(r)=\int_\eta^{+\infty}G_h(r,t)dt\,,\quad \delta_h(r)=e^{d(r)/h}u_h(r)-h^{-1/2}a_0(r)\]
and
\[\mathcal W_1(\eta)=h^{-1/2}\mathcal W_2(\eta)+\mathcal R_1(\eta)+\mathcal R_2(\eta)\]
where
\[
\begin{aligned}
\mathcal R_1(\eta)&=\big(C_h-C^{\rm asy}_h\big)\int_\eta^a r|v_0(r)|u_h(r) g_\eta(r)\exp\left( -\frac{r^2+L^2}{4h}\right)dr\,, \\
\mathcal R_2(\eta)&=C_h^{\rm asy}\int_\eta^a r|v_0(r)|\delta_h(r)g_\eta(r)\exp\left( -\frac{d(r)}{h}-\frac{r^2+L^2}{4h}\right)dr\,. 
\end{aligned}\]
Notice that 
\[\mathcal W_1(\eta)=C_h\int_\eta^a r|v_0(r)|u_h(r) g_\eta(r)\exp\left( -\frac{r^2+L^2}{4h}\right)dr\,.\]
Consequently 
\[ \mathcal R_1(\eta)=\left(1-\frac{C_h^{\rm asy}}{C_h}\right)\mathcal W_1(\eta)\underset{h\to0}{=}o\big(\mathcal W_1(\eta)\big)\,.\]
As for $\mathcal R_2(\eta)$, by  \eqref{eq:uh-main0}, $|\delta_h|=\mathcal O(h^{1/2})$, and  since $a_0(\cdot)>0$, we have
\[\mathcal W_2(\eta)\geq \mathcal W_2^0(\eta):=\mathfrak m_0 C_h^{\rm asy} \int_\eta^a r|v_0(r)| g_\eta(r)\exp\left( -\frac{d(r)}{h}-\frac{r^2+L^2}{4h}\right)dr \]
where
\[\mathfrak m_0=\min_{0\leq r\leq a}a_0(r)>0\,.\]
Consequently,  
\[| \mathcal R_2(\eta)|= \mathcal O\big(h^{1/2}\mathcal W_2^{0}(\eta)\big)=\mathcal O\big(h^{1/2}\mathcal W_2(\eta)\big )\,.\]
\end{proof}
 We give a finer approximation of the hopping coefficient $|w_{\ell,r}|$ in \eqref{eq:hat-w}, by replacing $G_h$ in the definition  of $\mathcal  W_2(\eta)$ (see Lemma~\ref{lem:exp-w*}) by its approximation at infinity. At  this  level, unlike Lemmas~\ref{lem:exp-w} and \ref{lem:exp-w*}, our estimates are no more uniform with respect  to $\eta\in(0,\eta_0]$.
\begin{lemma}\label{lem:exp-w**}
Let $\eta\in(0,\eta_0]$, where $\eta_0$ is introduced  in Lemma~\ref{lem:exp-w}. 
Consider
\[\mathcal W_3(\eta):= C_h^{\rm asy}\int_\eta^a r|v_0(r)|a_0(r)\exp\left( -\frac{d(r)}{h}-\frac{r^2+L^2}{4h}\right)\left(\int_{\eta}^{+\infty} G_h^{\rm asy}(r,t)\,dt\right)dr \,,\]
where $C_h^{\rm asy}$ is introduced in \eqref{eq:Ch},
\[G_h^{\rm asy}(r,t)=\sqrt{h}\,\,\dfrac{\exp\big( -\frac{(r^2+L^2)t}{2h} +\frac{Lr\sqrt{t(t+1)}}{h}-\alpha\ln\left(1+\frac1{t}\right)\big)}{(2\pi Lr)^{1/2}\,t^{5/4}(t+1)^{1/4}}\]
and $\alpha$ is introduced in \eqref{eq:def-alpha}.
Then, we have
\[ |w_{\ell,r}|\underset{h\to0}{\sim} h^{-1/2}\mathcal W_3(\eta)\,.\]
\end{lemma}
\begin{proof}
By Lemmas~\ref{lem:exp-w} and \ref{lem:exp-w*}, it  suffices to prove that $\mathcal W_2(\eta)\underset{h\to0}{\sim} \mathcal W_3(\eta)$. By \eqref{eq:def-Gh}, we observe that
\[ G_h(r,t)= \frac1{t}I_0\left(\frac{Lr\sqrt{t(t+1)}}{h}\right)\exp\left( -\frac{(r^2+L^2)t}{2h}-\alpha\ln\left(1+\frac1{t}\right)\right)\,.\]
By \eqref{eq:I0-main},  for each $\epsilon>0$, there exists $z_0>0$ such that
\[\forall\,z\geq z_0,\quad \left|I_0(z)-\frac{e^z}{\sqrt{2\pi z}}\right|<\epsilon  \frac{e^z}{\sqrt{2\pi z}}\,.\]
In particular, there exists
$h_0=h_0(\eta,\epsilon)>0$ such  that,
\[\forall\,r,t\geq\eta, ~\forall\,h\in(0,h_0],~|G_h(r,t)-G_h^{\rm asy}(r,t)|<\epsilon\, G_h^{\rm asy}(r,t)\,.\]
Writing
\[\mathcal W_2(\eta)=\mathcal W_3(\eta)+\mathcal R(\eta)\]
we get, for all $h\in(0,h_0]$,
\[ 
|\mathcal R(\eta)|\leq \epsilon \mathcal W_3(\eta)\]
which yields that $\mathcal R(\eta)\underset{h\to0}{=}o\big(\mathcal W_3(\eta)\big)$.
\end{proof}
In  our next step  we get  a new asymptotics from  Lemma~\ref{lem:exp-w**} by replacing $\alpha$ by  its approximation in \eqref{eq:def-alpha}. Recall that by  \eqref{eq:def-alpha}
\begin{equation}\label{eq:def-alpha*}
\begin{aligned}
\alpha&=\frac{\alpha_0}{h},\\
\alpha_0&=\alpha_0^{\rm main}+\mathcal O(h^{3/2})\,,\\
\alpha_0^{\rm main}&=\frac{|v_0^{\min}|}{2}-\frac{h}{2}\left(\sqrt{1+2v_0''(0)} -1\right)
\,.
\end{aligned}
\end{equation}
Inserting this into $G_h^{\rm asy}(r,t)$ in Lemma~\ref{lem:exp-w**}, we get  
\begin{equation}\label{eq:def-Gh-main}
\begin{aligned}
G_h^{\rm main}(r,t)&=\sqrt{\frac{h}{2\pi L r}}\,g_0(t)\,\exp\left( -\frac{(r^2+L^2)t}{2h} +\frac{Lr\sqrt{t(t+1)}}{h}-\frac{|v_0^{\min}|\ln\left(1+\frac1{t}\right)}{2h}\right)\\
g_0(t)&=\frac1{t^{5/4}(t+1)^{1/4}}\left(1+\frac1{t}\right)^{\frac12(\sqrt{1+2v_0''(0)}-1)}\,.
\end{aligned}
\end{equation} 
\begin{lemma}\label{lem:exp-w***}
Let $\eta\in(0,\eta_0]$, where $\eta_0$ is introduced  in Lemma~\ref{lem:exp-w}. 
Consider
\[\mathcal W_4(\eta):= C_h^{\rm asy}\int_\eta^a r|v_0(r)|a_0(r)\exp\left( -\frac{d(r)}{h}-\frac{r^2+L^2}{4h}\right)\left(\int_{\eta}^{+\infty} G_h^{\rm main}(r,t)\,dt\right)dr \,,\]
where $C_h^{\rm asy}$ is introduced in \eqref{eq:Ch},
$G_h^{\rm main}$ in \eqref{eq:def-Gh-main} and  $\alpha_0^{\rm main}$ in \eqref{eq:def-alpha*}.

Then, we have
\[ |w_{\ell,r}|\underset{h\to0}{\sim} h^{-1/2}\mathcal W_4(\eta)\,.\]
\end{lemma}
\begin{proof}
By Lemma~\ref{lem:exp-w**}, it  suffices to prove that $\mathcal W_3(\eta)\underset{h\to0}{\sim} \mathcal W_4(\eta)$. Notice that
\[G_h^{\rm asy}(r,t)=G_h^{\rm main}(r,t) \exp\left(\frac{\alpha_0^{\rm main}-\alpha_0}{h}\ln\left(1+\frac1{t}\right) \right)\]
and there exist $h_0=h_0(\eta)>0$ and $C_0>0$ such that, 
\[
 \forall\,h\in(0,h_0],~\forall\,t\geq \eta,~
 0\leq \frac{\alpha_0^{\rm main}-\alpha_0}{h}\ln\left(1+\frac1{t}\right)\leq C_0h^{1/2}\ln\left(1+\frac1\eta\right)\,.
 \]
 This proves that $G_h^{\rm  asy}(r,t)\underset{h\to0}{\sim}G_h^{\rm main}(r,t)$ uniformly with respect to 
 $(r,t)\in[\eta,a]\times[\eta,+\infty)$. This yields
 \[ \mathcal W_3(\eta)-\mathcal  W_4(\eta)\underset{h\to0}{=}o\big(\mathcal W_4(\eta)\big)\,.\] 
\end{proof}
\subsection{A real phase and its minimum}

 Using the expression of $C_h^{\rm asy}$ in  \eqref{eq:Ch}, we get by Lemma~\ref{lem:exp-w***}
\begin{equation}\label{eq:W4=exp-Psi}
\mathcal W_4(\eta)=\frac{\mathfrak  m(\mathfrak v_0)}{\sqrt{2\pi h}}\int_\eta^a \sqrt{r}\,|v_0(r)|a_0(r)\int_\eta^{+\infty} g_0(t)\exp\left(-\frac{\Psi(r,t)-F(\mathfrak v_0)}{h}\right)dtdr
\end{equation} 
where $\mathfrak m(\mathfrak v_0),F(\mathfrak v_0)$ are  introduced in \eqref{eq:def-m-f}, $g_0(t)$ is introduced in \eqref{eq:def-Gh-main} and 
\begin{equation}\label{eq:def-Psi}
\Psi(r,t):=d(r)+\frac{r^2+L^2}{4}(2t+1)+\frac{|v_0^{\min}|}{2}\ln\left(1+\frac1{t}\right)-Lr\sqrt{t(t+1)}\,.
\end{equation}
The Laplace like integral on the right hand side of \eqref{eq:W4=exp-Psi}  can be more precisely estimated in the limit $h\rightarrow 0$  once we know the infimum of $\Psi$ over $(\eta,a] \times (\eta,+\infty)$ of  $\Psi(r,t)$, which is the aim of the following  proposition.
\begin{proposition}\label{prop:min-Psi} If $L>2a$,  then
\[\Psi^{\min}:=\inf_{(r,t)\in[0,a]\times\R_+}\Psi(r,t)<\inf_{t\in\R_+}\Psi(0,t).\]
Furthermore, if \(v_0<0\) on \([0,a)\), then \(\Psi\) has a unique minimum, \((a,t_a)\), on \((0,a]\times\R_+\).
\end{proposition}

 The proof of Proposition~\ref{prop:min-Psi} relies on the following two lemmas.

\begin{lemma}\label{lem:min-Psi} Assume that $L>2a$. Let \((r_0,t_0)\in (0,L)\times \R_+\). Then
\((r_0,t_0)\) is a critical point of \(\Psi\) if, and only if, the following two conditions hold:
\begin{enumerate}
\item 
\begin{equation}\label{eq:def-ta}
t_{0}=t(r_0,L,v_0^{\min}):=\sqrt{\frac14+s(r_0,L,v_0^{\min})}-\frac12
\end{equation}
with
\begin{multline}\label{eq:def-s(a)}
s(r_0,L,v_0^{\min}):=\frac{2|v_0^{\min}|(L^2+r_0^2)+L^2r_0^2}{2(L^2-r_0^2)^2}\\+ \frac{1}{L^2-r_0^2}\sqrt{\frac{(2|v_0^{\min}|(L^2+r_0^2)+L^2r_0^2)^2}{4(L^2-r_0^2)^2} -|v_0^{\min}|^2}\,.
\end{multline}
\item  \(r_0\) is a solution of 
\begin{equation}\label{eq:new}
d'(r_0) + r_0 \sqrt{s(r_0,L,v_0^{\min})+\frac 14} - L \sqrt{s(r_0,L,v_0^{\min})}=0\,.
\end{equation}
\end{enumerate}
\end{lemma}

\begin{lemma}\label{lem:min-Psi2}
Let \(r_0\in(0,L)\). Then, \(r_0\) is a solution of \eqref{eq:new} if, and only if, \(v_0(r_0)=0\). 
\end{lemma}

\begin{proof}[Proof of Proposition~\ref{prop:min-Psi}]~\medskip

{\bf Step~1.}
The function $\Psi(r,t)$ depends continuously on  $(r,t)\in[0,a]\times\R_+$
and, using the inequality $\sqrt{t(t+1)}\leq t+\frac12$, we have  for all $(r,t)\in[0,a]\times\R_+$, 
\[\begin{aligned}
&\Psi(r,t) &\geq |v_0^{\min}| \ln\left (1+\frac1{t}\right) + \frac{(L-a)^2}{4}\,,\\
\mbox{ and} &&\\
&\Psi(r,t)&\geq \frac{L(L-2a)}2 t+\frac{L(L-2a)}{4}\,.\qquad .\end{aligned}\]
Hence, $\Psi^{\min}>0$ and there exist $\tau,T>0$ such that, for all $r\in[0,a]$ and $t\in(0,\tau]\cup [T,+\infty)$, we have $\Psi(r,t)>\Psi^{\min}$.  Thus, the minimum of $\Psi$  is  attained in $[0,a]\times(\tau,T)$.

Moreover, the minimum of  $\Psi$ can not be attained at $(0,t_0)$ with $t_0\in (\tau,T)$  since
\begin{equation*}
\partial_r \Psi (0,t_0) = - L \sqrt{t_0(t_0+1)} <0\,.
\end{equation*}

{\bf Step~2.} 
Pick \((r_0,t_0)\in(0,a]\times\R_+\) such that \(\Psi(r_0,t_0)=\Psi^{\min}\).  If \(r_0<a\) then \eqref{eq:new} holds, and by Lemma~\ref{lem:min-Psi2}, we get that \(v_0(r_0)=0\). This is impossible if we impose the hypothesis \(v_0<0\) on \([0,a)\), hence we have
\[r_0=a\mbox{ and }t_0=t(a,L,v_0^{\min}),\]
where
\(t(a,L,v_0^{\min})\) is introduced in \eqref{eq:def-ta}.
\end{proof}

\begin{proof}[Proof of Lemma~\ref{lem:min-Psi}]~\medskip

{\bf Step~1.} 
Let \((r_0,t_0)\in(0,L)\times\R_+\).  We will prove that 
\begin{equation}\label{eq:Psi-t=0}
\partial_t\Psi(r_0,t_0)=0\Longleftrightarrow t_0=t(r_0,L,v_0^{\min})
\end{equation}
where \(t_0(r_0,v_0^{\min})\) is introduced in \eqref{eq:def-ta}.

Starting with  \begin{equation}\label{eq:dt-Psi-def}
\partial_t \Psi(r_0,t):= \frac{r_0^2+L^2}{2} - \frac{|v_0^{\min}|}{2 t (t+1)} - \frac{ Lr_0 (2t+1)}{2\sqrt{t(t+1)}}
\end{equation}
we get 
\begin{equation*}
\partial_t\Psi(r_0,t)=\frac1{2s} g(s)\,,
\end{equation*}
where  $s = t^2+t >0$ and
\begin{equation*}
g(s)=  (L^2+r_0^2)s -|v_0^{\min}|-Lr_0 \sqrt{s} \sqrt{4s+1} \,.
\end{equation*}
The equation $g(s)=0$  reads
\begin{equation}\label{eq:dt-Psi=0}
 (L^2+r_0^2)s -|v_0^{\min}|= Lr_0 \sqrt{s(4s+1)} \,.
\end{equation}
This implies that a zero $\hat s$ of $g$ necessarily satisfies
\begin{equation}\label{neccond1}
 (L^2+r_0^2)\hat s -|v_0^{\min}| >0\,.
 \end{equation}
This  also implies by taking the square on both sides of \eqref{eq:dt-Psi=0},
\begin{equation} \label{neccond2}
 (L^2-r_0^2)^2 \hat s^2-\big(2(L^2+r_0^2)|v_0^{\min}|+L^2r_0^2 \big)\hat s +|v_0^{\min}|^2=0\,,
 \end{equation}
which has two positive solutions
\[s_\pm=\frac{2|v_0^{\min}|(L^2+r_0^2)+L^2r_0^2}{2(L^2-r_0^2)^2}\pm\frac{1}{L^2-r_0^2}\sqrt{\frac{(2|v_0^{\min}|(L^2+r_0^2)+L^2r_0^2)^2}{4(L^2-r_0^2)^2} -|v_0^{\min}|^2}\,. \]
 Notice that, if $g(\hat s)=0$, then
\[g'(\hat s)=  \frac{(L^2-r_0^2)^2}{  (L^2+r_0^2) \hat s -|v_0^{\min}|}\left(\hat s-\frac{2(L^2+r_0^2)|v_0^{\min}|+L^2r_0^2 }{2(L^2-r_0^2)^2} \right)\,.\]
Assuming that $s_\pm$ are zeros of $g$, we get
\begin{equation}\label{eq:g'(s-)}
g'(s_\pm)=  \pm \frac{L^2-r_0^2}{  (L^2+r_0^2) s_\pm  -|v_0^{\min}|}\, \sqrt{\frac{(2|v_0^{\min}|(L^2+r_0^2)+L^2r_0^2)^2}{4(L^2-r_0^2)^2} -|v_0^{\min}|^2}\,.
\end{equation}
At this stage we know, since $g(0)< 0$ and $\lim\limits_{s\to+\infty}g(s)=+\infty$, that  $g$ has at least one zero and at most two zeroes  which belong to $\{s_-,s_+\}$.
Suppose that $s_-$ is a zero of $g$, then we obtain by \eqref{neccond1} and \eqref{eq:g'(s-)} that $g'(s_-) <0$. This leads to a contradiction (or $g$ should have another zero $<s_-$).
Hence  the unique zero of $g$ is $s_+$. 

This yields, by \eqref{eq:dt-Psi-def}, that  $\partial_t\Psi(r_0,t)=0$ if, and only if, \(t\) satisfies
\begin{equation}\label{eq:formulepours}
\begin{aligned}
t^2+t&=s_+(r_0,L, v_0^{\min})\\
&:=\frac{2|v_0^{\min}|(L^2+r_0^2)+L^2r_0^2}{2(L^2-r_0^2)^2}+\frac{1}{L^2-r_0^2}\sqrt{\frac{(2|v_0^{\min}|(L^2+r_0^2)+L^2r_0^2)^2}{4(L^2-r_0^2)^2} -|v_0^{\min}|^2}\,.\end{aligned}
\end{equation}
Solving the previous equation, we end up with a unique positive solution
\[ t_0:= t_+(r_0,L, v_0^{\min})=-\frac12+\sqrt{\frac14+s_+(L,r_0,v_0^{\min}) } \,.\]
This proves \eqref{eq:Psi-t=0}.

{\bf Step~2.}

Assume now that \(t_0= t_+(r_0,L, v_0^{\min})\) and let us now prove  that, 
\begin{equation}\label{eq:Psi-r=0}
\partial_r\Psi(r_0,t_0)=0\Longleftrightarrow r_0\mbox{ satisfies  \eqref{eq:new}}.
\end{equation} 
If the partial derivative
\begin{equation}\label{eq:c-p-Psi}
\partial_r\Psi(r,t)=d'(r)+ \frac r 2 (2t+1) - L \sqrt{t(t+1)}
\end{equation}
vanishes, then
\[t=\hat t_\pm:=-\frac{L^2-r^2-2rd'(r)}{2(L^2-r^2)}\pm\sqrt{\left(\frac{L^2-r^2-2rd'(r)}{2(L^2-r^2)}\right)^2+\frac{(d'(r)+\frac{r}{2})^2}{L^2-r^2} }\,.\]
 Observing that $\displaystyle \hat t_- \hat t_+= -\frac{(d'(r)+\frac{r}{2})^2}{L^2-r^2}<0$ for $0<r<L$, we ignore the negative solution and
 consider
\begin{equation}\label{eq:def-t0(L,r)}
\hat t_+(r,L,v_0):=-\frac{L^2-r^2-2rd'(r)}{2(L^2-r^2)}+\sqrt{\left(\frac{L^2-r^2-2rd'(r)}{2(L^2-r^2)}\right)^2+\frac{(d'(r)+\frac{r}{2})^2}{L^2-r^2} }\,.
\end{equation}
Moreover, fixing \(r\in(0,L)\), \(f(t):=d'(r)+ \frac r 2 (2t+1) - L \sqrt{t(t+1)}\)  is \(>0\) for \(t=0\) and its limit is equal to \(-\infty\) as \(t\to-\infty\), so it must vanish at some \(t>0\). Therefore, we conclude that \(\partial_r\Psi(r,t)\) vanishes on \((0,L)\times\R_+\) if, and only if \(t=\hat t_+(r,L,v_0^{\min})\).

By \eqref{eq:Psi-t=0}, \((r_0,t_0)\) is a critical point of \(\Psi(r,t)\) if, and only if,
\begin{equation}\label{eq:t+=ht+}
t_0=\hat t_+(r_0,L,v_0)= t_+(r_0,L,v_0^{\min})\,.
\end{equation}
We will prove that \eqref{eq:t+=ht+} holds if, and only if \eqref{eq:def-ta} and \eqref{eq:new} hold. 
Firstly, notice that  \(t_0= t_+(r_0,L,v_0^{\min})\) is just \eqref{eq:def-ta} and \(\hat t_+(r_0,L,v_0)= t_+(r_0,L,v_0^{\min})\) is the same as
\begin{multline}\label{eq:t+=ht+2}
r_0d'(r_0) +(L^2-r_0^2)\sqrt{\left(\frac{L^2-r_0^2-2r_0d'(r_0)}{2(L^2-r_0^2)}\right)^2+\frac{(d'(r_0)+\frac{r_0}{2})^2}{L^2-r_0^2} }\\=(L^2-r_0^2)\sqrt{\frac14+s(r_0,L,v_0^{\min})},
\end{multline}
hence \eqref{eq:t+=ht+} is equivalent to assuming that \eqref{eq:def-ta} and \eqref{eq:t+=ht+2} both hold.

Secondly, if \eqref{eq:t+=ht+2} holds, then we can rewrite it in the form
\begin{multline*}(L^2-r_0^2)\sqrt{\left(\frac{L^2-r_0^2-2r_0d'(r_0)}{2(L^2-r_0^2)}\right)^2+\frac{(d'(r_0)+\frac{r_0}{2})^2}{L^2-r_0^2} }\\=(L^2-r_0^2)\sqrt{\frac14+s(r_0,L,v_0^{\min})}-r_0d'(r_0),
\end{multline*}
and after taking the squares we obtain
\[\left(d'(r_0)+r_0\sqrt{\frac14+s(r_0,L,v_0^{\min})}\right)^2-L^2s(r_0,L,v_0^{\min})=0, \]
which eventually yields  the relation in \eqref{eq:new}. Conversely, if \eqref{eq:def-ta} and \eqref{eq:new} hold, it is straightforward to check that,  for \(0<r_0<L\), we have
\[(L^2-r_0^2)\sqrt{s(r_0,L,v_0^{\min})+\frac14} -r_0d'(r_0)>L(L-r_0)\sqrt{s(r_0,L,v_0^{\min})}>0,\] 
and consequently \eqref{eq:t+=ht+2} holds.
\end{proof}
\begin{proof}[Proof of Lemma~\ref{lem:min-Psi2}]~\medskip

Let \(r\) be a solution of \eqref{eq:new} and \(s=s(r,L,v_0^{\min})\). Notice that, with 
\[\mu=|v_0^{\min}|,\]
\eqref{eq:def-s(a)} reads as follows
\begin{equation}\label{eq:s-new}
(L^2-r^2)^2s=\mu(L^2+r^2)+\frac{L^2r^2}{2}+2Lr\sqrt{\left(\frac{r^2}{4}+\mu\right)\left(\frac{L^2}{4}+\mu\right)}.
\end{equation}
We write \eqref{eq:new} in the following form
\[d'(r)=L\sqrt{s}-r\sqrt{s+\frac14}\]
then we take the square and get
\[v_0(r)+\mu=(L^2+r^2)s-2Lr\sqrt{s\left(s+\frac14\right)}\,.\]
After a rearrangement we have
\[v_0(r)+\mu-(L^2+r^2)s=-2Lr\sqrt{s\left(s+\frac14\right)}\,.\]
Taking the square, we obtain the following equation,
\[(v_0(r)+\mu)^2+(L^2+r^2)^2s^2 -2(v_0(r)+\mu)(L^2+r^2)s=4L^2r^2 s^2+L^2r^2s.\]
Rearranging the terms, we write the previous equation in the form
\[(L^2-r^2)^2s^2 -2\left((v_0(r)+\mu)(L^2+r^2)+\frac{L^2r^2}{2}\right)s+(v_0(r)+\mu)^2=0.\]
Multiplying by \((L^2-r^2)^2\), 
\begin{multline}\label{eq:critical-r}
\bigl((L^2-r^2)^2s\bigr)^2 -2\left(v_0(r)(L^2+r^2)+\mu(L^2+r^2)+\frac{L^2r^2}{2}\right)(L^2-r^2)^2s\\+(v_0(r)+\mu)^2(L^2-r^2)^2=0.\end{multline}
It is straightforward to check that
\begin{multline*}
\left(\mu(L^2+r^2)+\frac{L^2r^2}{2}+2Lr\sqrt{\left(\frac{r^2}{4}+\mu\right)\left(\frac{L^2}{4}+\mu\right)}\right)^2+\mu^2(L^2-r^2)^2\\=2\left(\mu(L^2+r^2)+\frac{L^2r^2}{2}\right)\left(\mu(L^2+r^2)+\frac{L^2r^2}{2}+2Lr\sqrt{\left(\frac{r^2}{4}+\mu\right)\left(\frac{L^2}{4}+\mu\right)}\right).
\end{multline*}
So, using \eqref{eq:s-new}, the above equation yields,
\[\bigl((L^2-r^2)^2s\bigr)^2-2\left(\mu(L^2+r^2)+\frac{L^2r^2}{2}\right)(L^2-r^2)^2s+\mu^2(L^2-r^2)^2=0.\]
Inserting this into \eqref{eq:critical-r}, we obtain
\[(L^2-r^2)^2v_0(r)^2-2(L^2+r^2)(L^2-r^2)^2s\,v_0(r)+2\mu(L^2-r^2)^2v_0(r)=0.\]
If \(v_0(r)\not=0\), we get
\[v_0(r)=2(L^2+r^2)s-2\mu\]
which is impossible since \(v_0\leq 0\) and the term on the right hand side above is positive, since we have by \eqref{eq:s-new},
\begin{multline*}
(L^2+r^2)(L^2-r^2)^2s-(L^2-r^2)^2\mu>\mu(L^2+r^2)^2-\mu(L^2-r^2)^2\\=4\mu L^2r^2>0.
\end{multline*}
\end{proof}

\begin{remark}[Global minimum of \(\Psi\)]\label{rem:glob-min}
Collecting Lemmas~\ref{lem:min-Psi} and \ref{lem:min-Psi2}, we see that the critical points of \(\Psi\) on \((a,L)\times\R_+\) constitute the set
\[\mathcal C=\{(r,t)~:~a<r<L,~t=t(r,L,v_0^{\min})\}.\]
Consequently, for all \(L'\in(a,L)\),
\[\min_{(r,t)\in[a,L']\times \R_+}\Psi(r,t)=\min\left(\inf_{t>0}\Psi(a,t),\inf_{t>0}\Psi(L',t),\inf_{r\in(a,L')}\Psi\bigl(r,t(r,L,v_0^{\min})\bigr)\right).\]
Moreover, for all \(r\in(a,L)\), we have by Lemma~\ref{lem:min-Psi},
\begin{multline*}\frac{d}{dr}\Psi\bigl(r,t(r,L,v_0^{\min})\bigr)\\=\underset{=0}{\underbrace{\partial_r\Psi\bigl(r,t(r,L,v_0^{\min})\bigr)}}+\bigl(\partial_r t(r,L,v_0^{\min})\bigr)\underset{=0}{\underbrace{\partial_t\Psi\bigl(r,t(r,L,v_0^{\min})\bigr)}}=0
\end{multline*}
hence
\[\Psi\bigl(r,t(r,L,v_0^{\min})\bigr)=\Psi\bigl(a,t(a,L,v_0^{\min})\bigr)=\Psi^{\min}\quad (a< r<L).\]
In particular, we have that
\[\min_{(r,t)\in[a,L-a]\times \R_+}\Psi(r,t)=\min\left(\Psi^{\min},\inf_{t>0}\Psi(L-a,t)\right)=\Psi^{\min},\]
and
\[\begin{aligned}
\Psi^{\min}&=\Psi\bigl(L-a,t(L-a,L,v_0^{\min})\bigr)\\
&=d(L-a)+\frac{a^2}{2}\sqrt{\frac14+s}+L(L-a)\left(\sqrt{\frac14+s}-\sqrt{s}\right)+\frac{|v_0^{\min}|}{2}\ln\frac{\bigl(\sqrt{\frac14+s}+\frac12\bigr)^2}{s},
\end{aligned}\]
where
\begin{multline}\label{eq:s(L-a)}
s=s(L-a,L,v_0^{\min})\\
=\frac{|v_0^{\min}|(L^2+(L-a)^2)+L^2(L-a)^2}{a^2(2L-a)^2}\\
+\frac{2L(L-a)}{a^2(2L-a)^2}\sqrt{\frac{(L-a)^2}{4}+|v_0^{\min}|}\sqrt{\frac{L^2}{4}+|v_0^{\min}|}\,.
\end{multline}
\end{remark}

\begin{remark}[Application to limit cases]\label{rem:asy-(L,a)}
For \(L>2a\),  let \(t(a,L,v_0^{\min})\mbox{ and }s(a,L,v_0^{\min})\) be as introduced in \eqref{eq:def-ta} and \eqref{eq:def-s(a)} respectively.
\begin{enumerate} 
\item
Assuming \(a\) and \(v_0^{\min}\) are fixed and \(L\to+\infty\),  we have
\[ t(a,L,v_0^{\min})\sim s(a,L,v_0^{\min})=\frac{|v_0^{\min}|+\frac{a^2}{2}+a\sqrt{\frac{a^2}{4}+|v_0^{\min}|}}{L^2}+\mathcal O\left(\frac{1}{L^4}\right),\]
hence 
Proposition~\ref{prop:min-Psi} yields
\[\Psi^{\min}=\frac{L^2}{4}+|v_0^{\min}|\ln L+\mathcal O(1).\]
\item 
Assuming \(L\) and \(v_0^{\min}\) are fixed and \(a\to0_+\),  we get  
\[s(a,L,v_0^{\min})\sim \frac{|v_0^{\min}|}{L^2},\quad t(a,L,v_0^{\min})\sim  \sqrt{\frac{1}{4}+\frac{|v_0^{\min}|}{L^2}}-\frac12,\]
hence  Proposition~\ref{prop:min-Psi} yields
\[\Psi^{\min}=\frac{L}4\sqrt{L^2+4|v_0^{\min}|}+|v_0^{\min}|\ln\left( \frac{L\Big(1+\sqrt{1+\frac{4|v_0^{\min}|}{L^2}}\Big)}{2\sqrt{|v_0^{\min}|}} \right)+o(1).\]
\end{enumerate}
\end{remark}
\subsection{Proof of main Theorem}
We give here the proof of Theorem~\ref{thm:main*}. We set
\begin{equation}\label{eq:def-S(v0)**}
S(\mathfrak v_0,L)=-F(\mathfrak  v_0)+\inf_{\substack{r\in[0,a]\\ t\in(0,+\infty)}}\Psi(r,t),
\end{equation}
where $a=a(\mathfrak v_0)$, $F(\mathfrak v_0)$ is introduced in \eqref{eq:def-m-f} and the function $\Psi$ is introduced in \eqref{eq:def-Psi}.  By Remark~\ref{rem:glob-min} we observe that
\begin{equation}\label{eq:exp-S(v0)}
S(\mathfrak v_0,L)=S_a+I(a,L,v_0^{\min})
\end{equation}
where \(S_a\) is introduced in \eqref{eq:Sa},
\begin{multline}\label{eq:def-I}
I(a,L,v_0^{\min})=\\\frac{a^2}{2}\left(\sqrt{\frac14+s}-\sqrt{\frac{1}{4}+\frac{|v_0^{\min}|}{a^2}} \right)
+\frac{|v_0^{\min}|}{2}\ln\frac{|v_0^{\min}|\bigl(\sqrt{\frac14+s}+\frac12 \bigr)^2}{a^2s\bigl(\sqrt{\frac14+\frac{|v_0^{\min}|}{a^2}}+\frac12\bigr)^2}
\end{multline}
and \(s=s(L-a,L,v_0^{\min})\) is introduced in \eqref{eq:s(L-a)}. It is easy to see that \(a^2s>|v_0^{\min}|\) so the term \(I(a,L,v_0^{\min})\) is positive.

Let $\eta\in(0,\eta_0)$, where $\eta_0\in(0,a)$ is introduced in Lemma~\ref{lem:exp-w}. 
So, by Lemma~\ref{lem:exp-w***}, it suffices to prove that,
\begin{equation}\label{eq:asy-W4}
h\ln \mathcal W_4(\eta)\underset{h\to0}{\sim} -S(\mathfrak v_0,L)\,.
\end{equation}
The function $v_0$ vanishes to infinite order at $r=a$.  By Proposition~\ref{prop:min-Psi}, we get
\begin{equation}\label{eq:ub-W4-sharp}
 \mathcal W_4(\eta)\underset{h\to0}{=}\mathcal O\left(e^{-\frac{S(\mathfrak v_0,L)}{h}}\right)\,,
 \end{equation}
from which it  follows
\begin{equation}\label{eq:W4-ub}
\limsup_{h\to0} \big(h\ln \mathcal W_4(\eta) \big)\leq -S(\mathfrak v_0,L)\,.
\end{equation}
To prove a lower bound on $\mathcal W_4(\eta)$, pick an arbitrary $\delta\in(0,1)$ and consider the set
\[I_{\delta}=\{(r,t)\in[0,a]\times[0,+\infty)~:~\Psi(r,t)\leq F(\mathfrak  v_0)+S(\mathfrak v_0,L)+ \delta\}\,.\]
By Proposition~\ref{prop:min-Psi}, there exists $\delta_0\in(0,1)$ such that, if $\delta<\delta_0$, then $I_{\delta}\subset[\eta,a]\times[\eta,+\infty)$.  Since the integrand in  the  expression of $\mathcal W_4(\eta)$ is positive,  we have the lower bound
\[ \mathcal W_4(\eta)\geq \frac{\mathfrak m(\mathfrak v_0)}{\sqrt{2\pi h}}\left( \int_{I_{\delta} }\sqrt{r}\,|v_0(r)|a_0(r)  g_0(t)dtdr\right) e^{-\frac{S(\mathfrak v_0,L)+\delta}{h}}\]
from which we infer the lower bound
\[\liminf_{h\to0} \big(h\ln \mathcal W_4(\eta)\big) \geq -S(\mathfrak v_0,L)-\delta\,.\]
After sending $\delta$ to $0$ we eventually get
\begin{equation}\label{eq:W4-lb}
\limsup_{h\to0} \big(h\ln \mathcal W_4(\eta)\big) \geq -S(\mathfrak v_0,L)\,.
\end{equation}
Collecting \eqref{eq:W4-ub} and \eqref{eq:W4-lb}, we finish the proof of \eqref{eq:asy-W4}.
\begin{remark}[Amplitude of the tunneling]\label{rem:amplitude}
Our proof  of Theorem~\ref{thm:main*} above only yields an  asymptotics for the phase of the hopping coefficient,
\begin{equation}\label{eq:asymp-w}
\ln |w_{\ell,r}|\underset{h\to0}{\sim} -S(\mathfrak v_0,L)\,.
\end{equation}
The initial hope was that we could estimate the amplitude of $|w_{\ell,r}|$ by showing that the function
$\Psi$  in \eqref{eq:def-Psi} has  a unique non-degenerate minimum $(r_*,t_*)$ on $[0,a]\times(0,+\infty)$ such  that $0<r_*<a$; if it were the  case, the method of Laplace approximation   would yield
\[ |w_{\ell,r}|\underset{h\to0}{\sim} \Upsilon(\mathfrak v_0) h^{-p(\mathfrak v_0)} e^{- S(\mathfrak v_0,L)/h}\,,\]
for some constants $\Upsilon(\mathfrak v_0)>0$ and $p(\mathfrak v_0)\in \R$. 
\end{remark}

The next proposition compares the magnetic and non-magnetic tunneling  asymptotics (see Remark~\ref{rem:b=0}).

\begin{proposition}\label{prop:beta=0}
We have $\lim\limits_{\beta\to0} \beta S(\beta^{-2}\mathfrak v_0)=2\displaystyle\int_0^{L/2}\sqrt{v_0(\rho)-v_0^{\min}}\,d\rho$.
\end{proposition}
\begin{proof}
Knowing that $S_a\leq S(\mathfrak v_0)\leq \hat S$ and replacing $v_0$ by $\beta^{-2}v_0$, we get that
\[S_a^\beta\leq S(\beta^{-2}\mathfrak v_0)\leq \hat  S^\beta \]
where (cf. \eqref{eq:Sa-*} and \eqref{eq:hat-S})
\[\begin{aligned}
S_a^\beta&=2\int_0^{a}\sqrt{\frac{\rho^2}{4}+\beta^{-2}\big( v_0(\rho)-v_0^{\min}\big)}\,d\rho+\int_a^{L-a}\sqrt{\frac{\rho^2}{4}-\beta^{-2} v_0^{\min}\big)}\,d\rho\,,\\
\hat S^\beta&=\inf_{0<r<a}\left(\frac{Lr}{2}+\int_0^{L-r}\sqrt{\frac{\rho^2}{4}+\beta^{-2}\big( v_0(\rho)-v_0^{\min}\big)}\,d\rho\right.\\
&\hskip3.5cm\left.+\int_0^r \sqrt{\frac{\rho^2}{4}+\beta^{-2}\big( v_0(\rho)-v_0^{\min}\big)}\,d\rho\right)\,.
\end{aligned}
\]
It is then clear that
\[
\begin{aligned}\lim_{\beta\to0}\beta S_a^\beta&=2\int_0^{a}\sqrt{  v_0(\rho)-v_0^{\min} }\,d\rho+\int_a^{L-a}\sqrt{|v_0^{\min}|}\,d\rho \,,\\and \\
\lim_{\beta\to0}\beta \hat  S^\beta&=\inf_{0<r<a}\left(\int_0^{L-r}\sqrt{ v_0(\rho)-v_0^{\min}}\,d\rho
+\int_0^{r}\sqrt{ v_0(\rho)-v_0^{\min}}\,d\rho\right)\\
&=2\int_0^{a}\sqrt{  v_0(\rho)-v_0^{\min} }\,d\rho+\int_a^{L-a}\sqrt{|v_0^{\min}|}\,d\rho\,.
\end{aligned}
\]
To conclude, we notice that 
\[\begin{aligned}
2\int_0^{L/2}\sqrt{  v_0(\rho)-v_0^{\min} }\,d\rho&=2\int_0^{a}\sqrt{  v_0(\rho)-v_0^{\min} }\,d\rho+2\int_a^{L/2}\sqrt{|v_0^{\min}|}\,d\rho\\
&=2\int_0^{a}\sqrt{  v_0(\rho)-v_0^{\min} }\,d\rho+\int_a^{L-a}\sqrt{|v_0^{\min}|}\,d\rho\,. 
\end{aligned}\]
\end{proof}

\subsection*{Acknowledgments}
This work was initiated while the second author  visited  LMJL at  the university of Nantes in November 2021 and February 2022.  The authors would like to acknowledge the  support from the  F\'ed\'eration de  Recherche  Math\'ematiques des Pays de Loire and Nantes Universit\'e. The first author benefits from the support of the French government ``Investissements d'Avenir''  program integrated to France 2030, bearing the following reference ANR-11-LABX-0020-01. The second author is partially supported by the Center for Advanced Mathematical Sciences (CAMS, AUB).
The first author thanks M. Weinstein for enlightening discussions around the first version of  \cite{FSW}.   The authors are thankful  to S. Fournais,  L. Morin and N. Raymond for pointing out   an error  appearing   in the proof of a former version of Proposition~6.5, and to M.P. Sundqvist for the attentive reading and discussions.

\end{document}